\documentclass[%
 reprint,
superscriptaddress,
 amsmath,amssymb,
 aps,
pra,
]{revtex4-2}

\usepackage{svg}
\usepackage[T1]{fontenc}
\usepackage{mathptmx}
\usepackage[colorlinks, linkcolor=red, anchorcolor=red, citecolor=red]{hyperref}
\usepackage{amsthm}
\usepackage{graphicx}
\usepackage{dcolumn}
\usepackage{bm}
\usepackage{qcircuit}
\usepackage{graphicx}
\usepackage{float}
\usepackage{subfig}
\usepackage{algpseudocodex}
\usepackage[linesnumbered, boxed, ruled]{algorithm2e}

\newtheorem{mydef}{Definition}
\newtheorem{mythe}{Theorem}
\newtheorem{mylem}{Lemma}
\newtheorem{mypro}{Problem}
\newtheorem{mycor}{Corollary}

\graphicspath{{figure/}}

\begin{document}

\preprint{APS/123-QED}

\title{An efficient quantum algorithm for independent component analysis}

\author{Xiao-Fan Xu}
\affiliation{CAS Key Laboratory of Quantum Information, University of Science and Technology of China, Hefei 230026, China}
\affiliation{CAS Center for Excellence in Quantum Information and Quantum Physics, University of Science and Technology of China, Hefei 230026, China}
\affiliation{Hefei National Laboratory, University of Science and Technology of China, Hefei 230088, China}
\author{Cheng Xue}
\thanks{xcheng@iai.ustc.edu.cn}
\affiliation{Institute of Artificial Intelligence, Hefei Comprehensive National Science Center, Hefei, Anhui, 230026, P. R. China}
\author{Zhao-Yun Chen}
\affiliation{Institute of Artificial Intelligence, Hefei Comprehensive National Science Center, Hefei, Anhui, 230026, P. R. China}
\author{Yu-Chun Wu}
\thanks{wuyuchun@ustc.edu.cn}
\affiliation{CAS Key Laboratory of Quantum Information, University of Science and Technology of China, Hefei 230026, China}
\affiliation{CAS Center for Excellence in Quantum Information and Quantum Physics, University of Science and Technology of China, Hefei 230026, China}
\affiliation{Hefei National Laboratory, University of Science and Technology of China, Hefei 230088, China}
\affiliation{Institute of Artificial Intelligence, Hefei Comprehensive National Science Center, Hefei, Anhui, 230026, P. R. China}
\author{Guo-Ping Guo}
\affiliation{CAS Key Laboratory of Quantum Information, University of Science and Technology of China, Hefei 230026, China}
\affiliation{CAS Center for Excellence in Quantum Information and Quantum Physics, University of Science and Technology of China, Hefei 230026, China}
\affiliation{Hefei National Laboratory, University of Science and Technology of China, Hefei 230088, China}
\affiliation{Institute of Artificial Intelligence, Hefei Comprehensive National Science Center, Hefei, Anhui, 230026, P. R. China}
\affiliation{Origin Quantum Computing Company Limited, Hefei, Anhui, 230026, P. R. China}

\begin{abstract}
Independent component analysis (ICA) is a fundamental data processing technique to decompose the captured signals into as independent as possible components. Computing the contrast function, which serves as a measure of independence of signals, is vital in the separation process using ICA. 
This paper presents a quantum ICA algorithm which focuses on computing a specified contrast function on a quantum computer. Using the quantum acceleration in matrix operations, we efficiently deal with Gram matrices and estimate the contrast function with the complexity of $O(\epsilon_1^{-2}\mbox{poly}\log(N/\epsilon_1))$. This estimation subprogram, combined with the classical optimization framework, enables our quantum ICA algorithm, which exponentially reduces the complexity dependence on the data scale compared with classical algorithms. The outperformance is further supported by numerical experiments, while a source separation of a transcriptomic dataset is shown as an example of application.



\end{abstract}
\maketitle
\section{\label{sec1.0}Introduction}

Independent component analysis (ICA) proposed during the 1990s \cite{BSS1991, ICA1992, ICA1994} has nowadays become one of the fundamental techniques for data processing. Belonging to matrix factorization methods including the prominent principal component analysis (PCA), ICA manages to decompose data points into independent ingredients based mainly on statistical properties. 
The main motivation for the development of ICA techniques is at first the blind source separation (BSS) problem aimed at recovering latent variables from the observations. Researchers from fields including acoustics \cite{ICAAppAcou1995, IcaAco2017}, telecommunication \cite{ICAAppTele}, image processing \cite{ICAAppImage1995, ICAFace2002}, biomedicine \cite{ICABio2004, ICAAppBio, ICABio2014, BioData2017} started to introduce and address the BSS problem in a variety of contexts.


With the cocktail-party problem \cite{ICAsum2000}, it's convenient for one to intuitively understand what a BSS problem is like. Imagine that two people are speaking simultaneously around you. Amplitudes of speech signals from these speakers are denoted by $s_1(t)$ and $s_2(t)$, where $t$ stands for time. You have two devices for recording. The recorded signal are then denoted by $x_1(t)$ and $x_2(t)$ which are naturally seen as a weighted sum of $s_1(t)$ and $s_2(t)$, written as
\begin{align}
    x_1(t)=a_{11}s_1+a_{12}s_2,\\
    x_2(t)=a_{21}s_1+a_{22}s_2.
\end{align}
The BSS problem in this case is to recover the coefficients $a_{ij}$, and $s_i(t)$ as well, from only the observations $x_i(t)$.

As the most widely used method for the BSS problem, ICA won great success since it is realistic for the basic assumption that the underlying processes are mutually independent when they correspond to distinct physical processes. Under this assumption, observable random variables $\bm{x}=(\bm{x}_1, \bm{x}_2,\dots)^T$ are arguably mixtures of some statistically independent "latent variables" $\bm{s}=(\bm{s}_1, \bm{s}_2,\dots)^T$. The linear mixture model is the most general choice which reads $\bm{x}=A\bm{s}$ with a coefficient matrix $A$. Algorithms that estimate $A$ as well as $\bm{s}$ came up successively mainly based on mutual information \cite{ICA1994, FastICA1997, FastICA1999, ICAMutual1999}, soon followed by algorithms involving more complex models such as linear convolutive mixtures \cite{ICAConv2003} or nonlinear models \cite{ICANonlin1999}.

An ICA algorithm typically proposes a contrast function that reflects the independence of a set of random variables, with the argument being their samples. The goal of the algorithm is to search for the linear mixtures of the samples that correspond to the extremum of the contrast function, perhaps by iteration or optimization. The scaling of any classical ICA algorithm is no less than the relevant problem size, which if we denote the sample size by $N$, can be referred to as $\Omega(N)$. 

As restricted computing resources inevitably fall short of demands suggested by practical applications with large amounts of data, it is natural to turn to the prospective quantum computation which has shown a great advantage over the classical counterpart \cite{HHL2009,VQE2017,HamsimQSP2017,QuanWalk2009,Shor1999,Quantum2002,Grov1997}. Furthermore, considering the strong correlation between ICA and PCA which can be exponentially accelerated by quantum PCA (qPCA) algorithm \cite{QPCA2014}, this kind of quantum algorithm promisingly serves as a subroutine for ICA to leverage this acceleration.

Contrast functions based upon statistics including kurtosis and mutual information \cite{ICA2002,ICAbook2010} can be directly estimated using existing quantum estimators with a certain speedup \cite{frame1998,EntropyEst2021}. However, kurtosis is in general not robust, and even an optimal entropy estimator cannot provide satisfactory quantum speedup. We simply give relevant results in Appendix.~\ref{app1}.

The Kernel ICA algorithm \cite{KICA2002} owns a contrast function defined by the principal components of the Gram matrices constructed by samples. An adapted version of this function can be efficiently estimated with our quantum algorithm introduced in this paper.
With given quantum access to the samples, 
our algorithm starts by encoding the Gram matrices onto amplitudes of a quantum state. By the parallel execution of two eigenvalue estimation circuits, the principal components of Gram matrices are extracted, measured, and used for the computation of the contrast function. As this step is the main bottleneck of the computation, it provides a critical speedup for the overall KICA. For sample size $N$, the complexity of performing ICA with our method scales as $O(\mbox{poly}\log(N))$, exponentially faster than its classical counterpart. 

The adaptation of the contrast function is to circumvent restrictions imposed by the quantum measurement. The problem may arise that the adapted function fails to give the correct independence of its input so that even when it approaches the extremum, separated results are not independent at all. However, we numerically verify its feasibility for source separation (including comparing the adapted and the original contrast function). We also give a rigorous error analysis for the output and explain that a linear separation of non-Gaussian variables is feasible. Numerical experiments are presented to verify the ability of separation and error performance of our algorithm. Furthermore, we numerically simulate the source separation process on a transcriptomic dataset as an illustration of applications.

The remainder of this paper is organized as follows: Section~\ref{sec2.0} lists the relevant conceptions of ICA and introduces the KICA algorithm. We describe our method adapted from KICA and present a theoretical analysis in Section~\ref{sec3.0}. 
Numerical experiments of our algorithm are shown in Section~\ref{sec4.0}. In Section~\ref{sec5.0} we give a conclusion and further discussion.

\section{\label{sec2.0}Independent component analysis: conceptions and algorithm}
In this section, we briefly recall conceptions of ICA and basic preprocessing techniques. Then we introduce the KICA algorithm as a background of our work.
\subsection{\label{sec2.1}ICA}
We take the linear ICA model for a rigorous definition of ICA \cite{ICAsum2000}. 
\begin{mydef}[ICA]
    Assume that $m$ observable random variables $\bm{x}_1,\bm{x}_2,\dots,\bm{x}_m$ can be modeled as linear mixtures of $m$ independent components (ICs) with real coefficients $A_{ij}$:
\begin{equation}
    \label{eq3}
    \bm{x}_j=A_{j1}\bm{s}_1+A_{j2}\bm{s}_2+\cdots+A_{jm}\bm{s}_m,
\end{equation}
where "IC" means these unknown $s_i$ are statistically independent. We use the notation that $\bm{x}=(\bm{x}_1,\bm{x}_2,\dots,\bm{x}_m)^T$, $\bm{s}=(\bm{s}_1,\bm{s}_2,\dots,\bm{s}_m)^T$ and matrix $\bm{A}$ with entries $A_{ij}$. ICA is to estimate $A$ together with ICs $\bm{s}$ with only samples of $\bm{x}$.
\end{mydef}
To avoid ambiguity, we denote the $j$-th of all $N$ samples of $\bm{x}_i$ by $x_{ij}$, and define vectors $x_i=(x_{i1},x_{i2},\dots,x_{iN})^T$ and matrix $X=(x_1,x_2,\dots,x_m)^T$, similarly for $s_{ij}, s_i, S$. Bold used or not is to distinguish between variables and samples.

To assess the independence of a set of random variables $\bm{x}_i$ given samples $X$, the so-called contrast function was put forward \cite{ICAbook2010}. In this paper, a contrast function is denoted by $J(X)$, with input samples $X$ and output reflecting the independence of these variables. Choices of contrast functions vary as different ICA algorithms are considered \cite{FastICA1997, ICAMutual1999, KICA2002}. A contrast function is kind of a separation criterion that determines how independent we perceive $\bm{x}_i$ to be. Given a linear transformation $\bm{W}$, $\bm{x}^{\prime}=\bm{Wx}$ is considered as more independent if $J(\bm{W}X)\leq J(X)$. It is naturally expected that the contrast function reaches its extremum when $\bm{W}=\bm{A}^{-1}$ so that a solution to ICA is obtained. A basic process of an ICA algorithm is to search for the optimal $\bm{W}$ according to $J(\bm{W}X)$ by optimization or iteration. In practice the optimized $\bm{W}$ is seen as an estimation of $\bm{A}^{-1}$, thus ICs $S=\bm{W}X$ are estimated.

Throughout this paper, we take the assumption that $x_i, s_i$ are real and share the same length for simplicity, and in addition, $\bm{s}$ are zero-mean and unit-variance. The independence demands the covariance matrix of $\bm{s}$ equal the identity. As $\bm{s}$ are zero mean, this can be written as $E\{\bm{s}\bm{s}^T\}=\bm{I}$, where $E\{\cdot\}$ represents expectation and $E\{\bm{s}\bm{s}^T\}$ is the matrix with entries $E\{s_is_j\}$. The assumption concretely reads 
\begin{equation}
    E\{s_is_j\}=\begin{cases}
        0, & i=j, \\
        1, & i\neq j.
        \end{cases}
\end{equation}
That the inverse of coefficient matrix $\bm{A}^{-1}$ exists is also assumed, so that we can write
\begin{equation}
    \bm{s}=\bm{A}^{-1}\bm{x}
\end{equation}
for later discussion.

\subsection{\label{sec2.2}Preprocessing}
Considering preprocessing before applying the ICA algorithm always makes the problem simpler and better conditioned. Especially, techniques like centering and whitening are almost necessary steps for ICA. 

Centering means making samples $x_i$ zero-mean by subtracting its expectation $x_{ij}-E\{\bm{x}_i\}$, where the expectation is in general estimated by the sample mean. Though Eq.~{eq3} has implied that $\bm{x}_j$ is zero-mean, by whitening ICA also works in the non-zero mean case. 

Suppose $\bm{x}_i$ itself has been centered, whitening is to transform $X$ into samples of 'whitened' variables $\bm{y}$ so that the covariance matrix of $\bm{y}$ equals the identity matrix: $E\{\bm{y}\bm{y}^T\}=\bm{I}$. Whitening is always possible by computing the eigendecomposition of the covariance matrix: $E\{\bm{xx}^T\}=\bm{EDE}^T$, where $\bm{E}$ and $\bm{D}$ are matrices composed of eigenvectors and eigenvalues, respectively. Then it is easy to check that
\begin{equation}
    \bm{y}=(E\{\bm{xx}^T\})^{-1/2}\bm{x}=\bm{ED}^{-1/2}\bm{E}^T\bm{x}=\bm{ED}^{-1/2}\bm{E}^T\bm{A}\bm{s},
\end{equation}
leads to
\begin{equation}
    E\{\bm{y}\bm{y}^T\}=\bm{I}.
\end{equation}
By assuming that $E\{\bm{s}\bm{s}^T\}=\bm{I}$, it's evident that the coefficient matrix $\bm{ED}^{-1/2}\bm{E}^T\bm{A}$ is orthogonal.

Substitute these statistics by sample mean:
\begin{equation}
    Y=\bm{ED}^{-1/2}\bm{E}^TX=\bm{ED}^{-1/2}\bm{E}^T\bm{A}S,
\end{equation}
where we can regard $Y$ as samples of $\bm{y}$. Note that $Y$ differs from $S$ by only an orthogonal transformation. Therefore, whitening makes it suffices to search through orthogonal matrices only for the optimal $\bm{W}=(\bm{ED}^{-1/2}\bm{E}^T\bm{A})^{-1}$ if we replace $X$ by $Y$.

\subsection{\label{sec2.3}KICA}
Among all the ICA algorithms, KICA, which we will introduce in a second, appears to be numerically relatively robust \cite{KICA2002} and related closely to matrix operations. As introduced before, KICA minimizes its contrast function for $\bm{W}$ by optimization for estimation of $\bm{A}^{-1}$.

The contrast function is determined and can be computed after selecting a kernel function $K(x,y)$, e.g. a Gaussian kernel we use later:
\begin{equation}
    K(x,y)=\exp\{\frac{-\left|x-y\right|^2}{2\sigma^2}\}.
    \label{eq5}
\end{equation}
The contrast function is computed by the following four steps, which can be considered roughly as its definition.

Step 1: Compute $K_i$ with selected kernel function $K(x,y)$ and samples $X$:
\begin{align}
    (\tilde{K}_i(X))_{jk}&=K(x_{ij},x_{ik}),\\
    K_i(X)=(\bm{I}-&\bm{1}/N)\tilde{K}_i(\bm{I}-\bm{1}/N),
    \label{eq14}
\end{align}
where $\tilde{K}_i$ ($K_i$) is called (centered) Gram matrix, respectively. Variables in parentheses specify the input samples. In this paper, the Gram matrices we used generally refer to the centered one in the absence of a specific designation. We denote the matrix with all entries 1 by $\bm{1}$ and the identity matrix by $\bm{I}$.

Step 2: Perform an eigenvalue decomposition for all Gram matrices:
\begin{equation}
    K_i(X)=U_i\Lambda_iU_i^T.
    \label{eq8}
\end{equation}
Denote the $k$-th eigenvalue of $K_i$ by $\lambda_{ik}$ and corresponding eigenvectors by $\vec{u}_{ik}$. $U_i$ and $\Lambda_i$ are defined as $U_i=(\vec{u}_{i1},\vec{u}_{i2},\dots,\vec{u}_{iN})$ and $\Lambda_i=\mbox{diag}(\lambda_{i1},\lambda_{i2},\dots,\lambda_{iN})$, respectively. 

Based upon the conclusion that only $M_i=O(\log N)$ largest eigenvalues make sense for Gaussian kernel \cite{KICA2002,nystrom2005}, an approximated decomposition can be used instead which is given by $K_i\approx\tilde{U}_i\tilde{\Lambda}_i\tilde{U}_i^T$, for $\tilde{U}_i=(\vec{u}_{i1},\vec{u}_{i2},\dots,\vec{u}_{iM_i})$ and $\tilde{\Lambda}_i=\mbox{diag}(\lambda_{i1},\lambda_{i2},\dots,\lambda_{iM_i})$. 

Step 3: Choose a positive constant $\kappa$, and then compute the positive-definite matrix $\mathcal{R_{\kappa}}$:
\begin{equation}\mathcal{R_{\kappa}}(X)=
    \begin{pmatrix}
        I&R_1U_1^TU_2R_2&\cdots&R_1U_1^TU_mR_m\\
        R_2U_2^TU_1R_1&I&\cdots&R_2U_2^TU_mR_m\\
        \vdots&\vdots&\ddots&\vdots\\
        R_mU_m^TU_1R_1&R_mU_m^TU_2R_2&\cdots&I
    \end{pmatrix},
    \label{eq16}
\end{equation}
where $R_i=$diag$(\lambda_{i1}^{\prime},\lambda_{i2}^{\prime},\dots,\lambda_{iM_i}^{\prime})$ with the notation that $\lambda_{ik}^{\prime}=\lambda_{ik}/(\lambda_{ik}+\kappa)$.

Step 4: Compute the determinant of $\mathcal{R}_{\kappa}$, $\det(\mathcal{R}_{\kappa})$, and contrast function is taken as $J(X)=-\ln(\det(\mathcal{R}_{\kappa}))$. 

We have introduced the overall computational process for contrast function $J(X)$ hereinbefore. The only job left is to perform optimization with $J(X)$. Moreover, as we have discussed in the last section, by replacing $X$ with whitened $Y$, the problem is reduced to an optimization for an orthogonal matrix, which is the optimization problem on the Stiefel manifold with relatively mature techniques to deal with \cite{OptiStieMani1998}.

\section{\label{sec3.0}Quantum algorithm for adapted KICA}

ICA algorithms like KICA, based mainly on matrix operations, may achieve better quantum speedup. Different from the original KICA, we replace the contrast function by $R_{\kappa}$ to circumvent the restrictions of measurements, where $R_{\kappa}$ is defined as the element-wise absolute values of $\mathcal{R}_{\kappa}$. We introduce a quantum algorithm to compute $J(\bm{W}Y)=\det(R_{\kappa}(\bm{W}Y))$ in this section. We use the following notations that $X=(x_1,x_2,\dots,x_m)^T$, $X'=(x'_1,x'_2,\dots,x'_m)^T$, $Y=(y_1,y_2,\dots,y_m)^T$ respectively represents the original, centered and preprocessed samples and another one that $\bm{W}Y=(y'_1,y'_2,\dots,y'_m)^T$.

\subsection{\label{sec3.1}Main results}
The problem to be solved on a quantum computer is stated in Problem~\ref{pro2}:
\begin{mypro}[Computation of contrast function]
    \label{pro2}
    Given $N$ samples $x_i$ of each of the $m$ random variables $\bm{x}_i$. Suppose we have the oracle access to $x_i$: $O_{x_i}|j\rangle|0\rangle=|j\rangle|x_{ij}\rangle$, where $|x_{ij}\rangle$ is an s-bit binary string representing $x_{ij}$ by qubits. The target is to estimate $\det(R_{\kappa}(\bm{W}Y))$ for some orthogonal $\bm{W}$ with only elementary gates and queries to $O_{x_i}$.
\end{mypro}

Denote the centered kernel by
\begin{equation}
    \begin{split}
    K'(y,x)=&K(y,x)-\int_{\mathbb{R}}d\rho_{z_i}(\alpha) K(\alpha,x)-\\
    &\int_{\mathbb{R}}d\rho_{z_i}(\beta) K(y,\beta)+\int_{\mathbb{R}^2}d\rho_{z_i}(\alpha)d\rho_{z_i}(\beta)K(\alpha,\beta).
    \end{split}
    \label{eq34new}
\end{equation}
In our analysis, we require $\bm{y}_i'$ to satisfy
\begin{equation}
    \int_{\mathbb{R}}K'(y_{im}',x)g(x)f_{y_i'}(x)dx=\frac{1}{N}\sum_{n=1}^N(K_i)_{mn}g(y_{im}'),
    \label{eq35new}
\end{equation}
for all $i,m$, arbitrary $\bm{W}$ and any function $g\in L^2(\mathbb{R})$, where $f_{y_i'}(x)$ is the probability density function (PDF) of $y_i'$. Under the commonly used assumption that the distribution of $\bm{x}_i$ is exactly one of obtained samples, this condition is naturally satisfied. 

The oracle $O_{x_i}$ may be provided by qRAM \cite{QRAM2008}, which allows this state preparation within $O(\mbox{poly}\log N)$ operations. Using oracles, we estimate the covariance matrix $M$ of $\bm{X}$ and then execute the preprocessing subprogram, where the estimation is denoted by $\tilde{M}$. This causes the input error so that the final output estimates $R_{\kappa}(\bm{W}\tilde{Y})$ instead. Denote by $\mu_M$ the condition number of $M$. The input error can then be expressed as
\begin{equation}
    \tilde{y}_{ij}=\sum_{k=1}^{m}(\tilde{M}^{-\frac{1}{2}})_{ik}(x_{kj}-E\{\bm{x}_k\}),
    \label{eq17}
\end{equation} 
and
\begin{equation}
    \|\tilde{M}^{-\frac{1}{2}}-M^{-\frac{1}{2}}\|_2<\epsilon_2.
    \label{eq16.5}
\end{equation}

The main result of our method for Problem~\ref{pro2} is stated as Theorem~\ref{the1}.
\begin{mythe}
    \label{the1}
    Denote by $d=\Theta(m\log N)$ and $\xi$ the dimension and the minimal eigenvalue of $R_{\kappa}(\bm{W}\tilde{Y})$. Given $\epsilon_1<\frac{1}{d^2}, \epsilon_2<0.2$ and $\kappa$, and samples $x_i$ with size $N$ of $m$ random variables $\bm{x}_i$ together with oracles access to $x_i$, $O_{x_i}|j\rangle|0\rangle=|j\rangle|x_{ij}\rangle$, there exists a quantum algorithm outputting $\det(\tilde{R}_{\kappa}(\bm{W}\tilde{Y}))$, where
    \begin{equation}
        \frac{|\det(\tilde{R}_{\kappa}(\bm{W}\tilde{Y}))-\det(R_{\kappa}(\bm{W}\tilde{Y}))|}{\det(R_{\kappa}(\bm{W}\tilde{Y}))}=\tilde{O}(\epsilon_1),
    \end{equation}
    using $\tilde{O}(\frac{1}{\xi^2\epsilon_1^{2}}+\frac{\mu_M^2}{\epsilon_2\|M\|_2^{3/2}})$ elementary gates and queries to each $O_{x_i}$.
    \end{mythe}
We neglect factors in $m,\kappa$ and use $\tilde{O}$ to denote the complexity up to polylogarithmic factors in $N, \epsilon_1, \epsilon_2, \mu_E$. It is worth noting that $\xi$ is usually $O(1)$ during our tests when appropriately choosing $\kappa$, for example, to be $0.1$.


Further analysis of the special case is given in Theorem~\ref{the2}, that $\tilde{\bm{y}}'=\bm{W}\tilde{\bm{y}}$ has few mutual dependence (called near-independent case later), where a better complexity can be achieved.

\begin{mythe}
    \label{the2}
    Continue with the notation of Theorem 1, but require that $\epsilon_1<1, \epsilon_2^2\ll\epsilon_2=o(d^{-2})$ instead. When 
    \begin{equation}
        \tilde{\bm{y}}'=(1+\epsilon_2F)\bm{s}+o(\epsilon_2^2),
    \end{equation}
    where $\|F\|_{max}\leq1$ is assumed, the precision of output can be changed into:
    \begin{equation}
        \frac{|\det(\tilde{R}_{\kappa}(\bm{W}\tilde{Y}))-\det(R_{\kappa}(\bm{W}\tilde{Y}))|}{\det(R_{\kappa}(\bm{W}\tilde{Y}))}=\tilde{O}(\epsilon_1\epsilon_2),
    \end{equation}
    with the query and gate complexity of $\tilde{O}(\frac{1}{\epsilon_1^{2}\epsilon_2}+\frac{\mu_M^2}{\epsilon_2\|M\|_2^{3/2}})$.
\end{mythe}

The main idea of our method is that, we first construct the circuit $O_{y'_i}|j\rangle|0\rangle=|j\rangle|y'_{ij}\rangle$ by a quantum preprocessing subprogram, then manage to encode the Gram matrix $K_i(\bm{W}Y)$ onto a quantum state, and finally by quantum circuits 
and quantum measurements, we obtain the classical information to reconstruct the low dimensional $R_{\kappa}(\bm{W}Y)$ together with its determinant. The preprocessing subprogram and function estimation are respectively introduced in Section~\ref{sec3.2} and Section~\ref{sec3.3}. We summarize our method for the near-independent case in Algorithm~\ref{alg1}, and the general one can be obtained by slight modification. A flow chart is shown in Fig.~\ref{flowchart} including optimization discussed in Section~\ref{sec3.4}.

\begin{figure*}[ht]
    \centering
    \includegraphics[height=13cm,width=13cm]{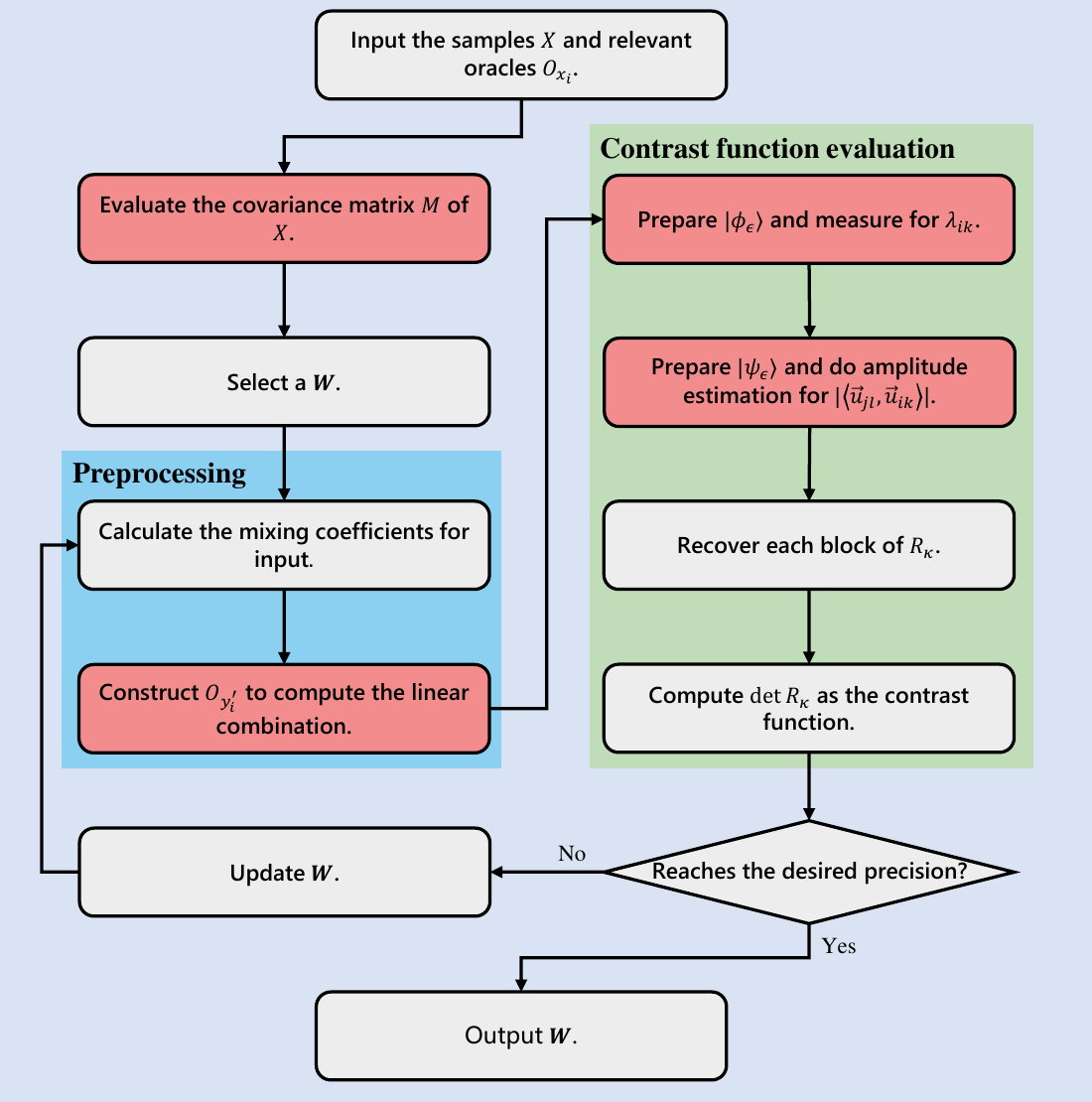}
    \captionsetup[subfloat]{labelsep=none, format=plain, labelformat=empty}
    \caption{\label{flowchart} A flowchart outlining the process for optimizing to obtain the optimal unmixing coefficients $\bm{W}$, where the red block implies the procedure using a quantum computer.}
\end{figure*}

\begin{algorithm}[htb]
    \textbf{Input:} Near-independent samples $X$, Sample size $N$, state preparation oracles $O_{x_i}$, parameterized matrix $\bm{W}$, kernel function $K(x,y)$, preselected parameter $\kappa$, error tolerance $\epsilon_1$, $\epsilon_2$.

    \textbf{Output:} $\det(\tilde{\mathcal{R}}_{\kappa}(\bm{W}\tilde{Y}))$ with multiplicative error of $\tilde{O}(\epsilon_1\epsilon_2)$ and $\sqrt{\sum_{i=1}^{m}(\tilde{y}'_{ij}-y'_{ij})^2}<\epsilon_2$.

    \textbf{Runtime:} $\tilde{O}(\frac{1}{\epsilon_1^2\epsilon_2}+\frac{\mu_M^2}{\epsilon_2\|M\|_2^{3/2}}))$ uses of oracles $O_{x_i}$ and its inverse, and elementary quantum gates.

    \textbf{Procedure:}\;
    Compute the covariance matrix $M$ of $\bm{x}$\;
    Use $M$ to construct the state preparation $O_{\tilde{y}'_i}$ for every $i$, and set $\bm{z}=\tilde{\bm{y}}'$\;
    \ForEach{$0\leq i<m$}{
        \For{$k=1,2,\dots,K=\tilde{O}(1/\epsilon_1)$}{
            Use the circuit $O_{\epsilon_1}$ to prepare $|\phi_{\epsilon_1}\rangle$\;
            Measure the last register for $\lambda_{ik}/N$\;
        }
    } 
    \ForEach{$0\leq i<j<m$}{
        \ForEach{$\tilde{\lambda}_{ik}$ and $\tilde{\lambda}_{jl}$ measured}{
            Use the amplitude estimation with $U_{\epsilon_1}$ for overlaps between $|\psi_{\epsilon_1}\rangle$ and $|\tilde{\lambda}_{ik},\tilde{\lambda}_{jl}\rangle$\;
        }
    }
    Calculate the contrast function using the measured values.
\caption{Computation of the contrast function.\label{alg1}}
\end{algorithm}

\subsection{\label{sec3.2}Preprocessing on a quantum computer}
The preprocessing subprogram has three steps: Step 1: use the quantum mean estimator for the covariance matrix $M$ of $\bm{x}$; step 2: classically compute the mixing coefficients, $\bm{W}M^{-1/2}$; step 3: use controlled weighted sum algorithm introduced in \cite{QuanArith2017} to linearly transform $O_{x_i}$ into $O_{y_i'}$. The formal result is stated in Lemma~\ref{lem11}.
\begin{mylem}
    \label{lem11}
    Given oracles $O_{x_i}:|j\rangle|0\rangle\to|j\rangle|x_{ij}\rangle$ when $\max_{i,j}\{|x_{ij}-E\{\bm{x}_i\}|\}\leq1$ is assumed, sample means $E\{x_i\}$ and the covariance matrix $M$ can be estimated within $\tilde{O}(\frac{\mu_M^2m}{\|M\|_2^{3/2}\epsilon})$ gates and queries to $O_{x_i}$, so that the estimated $\tilde{M}^{-\frac{1}{2}}$ satisfies that
    \begin{equation}
        \|\tilde{M}^{-\frac{1}{2}}-M^{-\frac{1}{2}}\|_2<\epsilon,
    \end{equation}
    where $\mu_M$ is the condition number of $M$. With such estimation, $O_{y'_i}$ can be executed with $O(\mbox{poly}\log(\frac{N}{\epsilon}))$ gates and queries to $O_{x_i}$, where
    \begin{equation}
        O_{y'_i}|j\rangle|0\rangle=|j\rangle|\tilde{y}'_{ij}\rangle.
        \label{eq16real}
    \end{equation}
    We have 
    \begin{equation}
        \tilde{y}'_{ij}=\sum_{k=1}^{m}(\bm{W}\tilde{M}^{-\frac{1}{2}})_{ik}(x_{kj}-\tilde{E}\{\bm{x}_k\}).
        \label{eq18}
    \end{equation}
\end{mylem}


\begin{proof}
The preprocessing has 2 phases, where the former estimates the covariance matrix and the latter utilizes the acquired knowledge for state preparation. 

For the first phase, we use a quantum mean estimator proposed in \cite{frame1998}, which, for example, outputs the mean of $x_i$ using oracles such as $O_{x_i}$ with gate and query complexity of $\tilde{O}(1/\epsilon)$. Ignore the error induced by quantum arithmetics \cite{QuanArith2017}, by which we can construct $O_{x_{i},x_{j}}:|k\rangle|0\rangle\to|k\rangle|x_{ik}x_{jk}\rangle$. With $O_{x_i}$ and $O_{x_{i},x_{j}}$, both $E\{x_{i}\}$ and $E\{x_{i}x_j\}$ are estimated with accuracy of $\frac{\|M\|_2^{3/2}\epsilon}{8\mu_M^2m}$. Therefore, $M$ can be estimated since $Cov(x_i,x_j)=E\{x_{i}x_j\}-E\{x_{i}\}E\{x_{j}\}$. The gate and query complexity of the computation of coefficients is eventually $\tilde{O}(\frac{\mu_M^2m}{\|M\|_2^{3/2}\epsilon})$ \cite{frame1998}.

Thus, calculate the coefficient matrix $\bm{W}M^{-1/2}$ classically using acquired $M$. Through the process of controlled weighted sum \cite{QuanArith2017} to linear combine $x_i$ on a quantum register, 
$O_{y_i'}$ in Eq.~\ref{eq16real} is constructed. Both query and gate complexity of $O_{y_i'}$ is $O(\mbox{poly}\log(N/\epsilon))$. The error bound of $M$ is proved in Appendix~{\ref{app2.1}}.

\end{proof}

\subsection{\label{sec3.3}Evaluation of the contrast function}
Here we start to introduce a quantum circuit computing $J(Z)$ for general input more than just $\bm{W}Y$ with queries to oracles $O_{z_i},i=1,2,\dots,m$. We use five subprocedures to achieve this goal. The first one is to construct an oracle $\tilde{O}_{K_i}$ outputting entries of uncentered Gram matrix $\tilde{K}_i(Z)$, where
\begin{equation}
    \tilde{O}_{K_i}|j\rangle|k\rangle|z\rangle\to|j\rangle|k\rangle|z\oplus(\tilde{K}_i)_{jk}\rangle.
\end{equation}

The second is to block encode the centered Gram matrix $K_i(Z)$, with the corresponding circuit named $U_{K_i}$. 
With $U_{K_i}$, an eigenvalue estimation circuit $U_{i}:|u_{ik}\rangle|0\rangle\to|u_{ik}\rangle|\lambda_{ik}/N\rangle$ based on phase estimation can be constructed, where 
\begin{equation}
    |u_{ik}\rangle=\sum_{\alpha=1}^N(\vec{u}_{ik})_{\alpha}|\alpha\rangle,
\end{equation}
and $\lambda_{ik}$, $\vec{u}_{ik}$ are $k$-th eigenvalue, normalized eigenvector of $K_i(Z)$. The third subprocedure is to implement $U_{i}$ and $U_{j}$ parallelly on the initial state $|K_i\rangle|0\rangle|0\rangle$, where
\begin{equation}
    |K_i(Z)\rangle=\frac{1}{N}\sum_{j,k=1}^N(K_i)_{jk}|j\rangle|k\rangle.
\end{equation}
shown in Fig.~\ref{fig2}. The output is about
\begin{equation}
    \sum_{k,l=1}^{N}\frac{\lambda_{ik}}{N}\langle \vec{u}_{jl}, \vec{u}_{ik}\rangle|u_{ik}\rangle|u_{jl}\rangle|\lambda_{ik}/N\rangle|\lambda_{jl}/N\rangle+|\perp\rangle,
\end{equation}
for $|\perp\rangle$ being the state of no interest to us. By quantum measurement and amplitude estimation, we obtain the values of $\lambda_{ik}/N$ and $|\langle \vec{u}_{jl}, \vec{u}_{ik}\rangle|$ and then compute $R_{\kappa}(Z)$ by
\begin{equation}
    (R_{\kappa})_{ik,jl}=\frac{\lambda_{ik}/N}{\lambda_{ik}/N+\kappa/2}\frac{\lambda_{jl}/N}{\lambda_{jl}/N+\kappa/2}|\langle\vec{u}_{ik},\vec{u}_{jl}\rangle|.
    \label{eq13}
\end{equation}
The process of measurement is called the fourth step while the classical computation of $R_{\kappa}(Z)$ as well as its determinant $J(Z)$ is called the fifth.

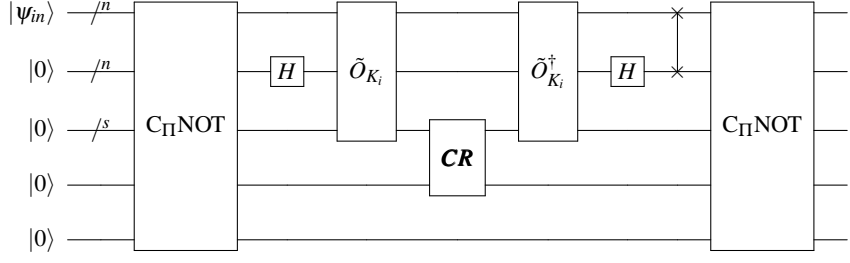
\begin{figure*}[hbt]
    \centerline{
    \Qcircuit @C=1.4em @R=1.4em {
    \lstick{|\psi_{in}\rangle}&{/^n}\qw&\multigate{4}{\mbox{C}_{\Pi}\mbox{NOT}}&\qw     &\multigate{2}{\tilde{O}_{K_i}}&\qw                   &\multigate{2}{\tilde{O}^{\dagger}_{K_i}}&\qw     &\qswap &\multigate{4}{\mbox{C}_{\Pi}\mbox{NOT}}&\qw \\
    \lstick{|0\rangle}&{/^n}\qw&\ghost{\mbox{C}_{\Pi}\mbox{NOT}}&\gate{H}&\ghost{\tilde{O}_{K_i}}       &\qw                   &\ghost{\tilde{O}^{\dagger}_{K_i}}       &\gate{H}&\qswap\qwx&\ghost{\mbox{C}_{\Pi}\mbox{NOT}}&\qw\\
    \lstick{|0\rangle}&{/^s}\qw&\ghost{\mbox{C}_{\Pi}\mbox{NOT}}&\qw     &\ghost{\tilde{O}_{K_i}}       &\multigate{1}{\bm{CR}}&
    \ghost{\tilde{O}^{\dagger}_{K_i}}       &\qw     &\qw       &\ghost{\mbox{C}_{\Pi}\mbox{NOT}}&\qw\\
    \lstick{|0\rangle}&\qw  &\ghost{\mbox{C}_{\Pi}\mbox{NOT}}   &\qw     &\qw                           &\ghost{\bm{CR}}       &
    \qw                                     &\qw     &\qw    &\ghost{\mbox{C}_{\Pi}\mbox{NOT}}&\qw\\
    \lstick{|0\rangle}&\qw   &\ghost{\mbox{C}_{\Pi}\mbox{NOT}}  &\qw    &\qw                       &\qw       &
    \qw                                     &\qw     &\qw    &\ghost{\mbox{C}_{\Pi}\mbox{NOT}}&\qw\\
    }
    }
    \caption{\label{fig71}The block-encoding circuit $U_{K_i}$ of Gram matrix $K_i$, which outputs $\frac{K_i}{N}|\psi_{in}\rangle$ when the last $n+s+2$ qubits are measured onto $|0^{\otimes n+s+1},1\rangle$. $\bm{CR}$ represents the controlled rotation: $|a\rangle|0\rangle\to|a\rangle(a|0\rangle+\sqrt{1-a^2}|1\rangle)$.} 
\end{figure*}

\begin{figure}[ht]
    \centering
    \includegraphics[height=5.7cm,width=8.5cm]{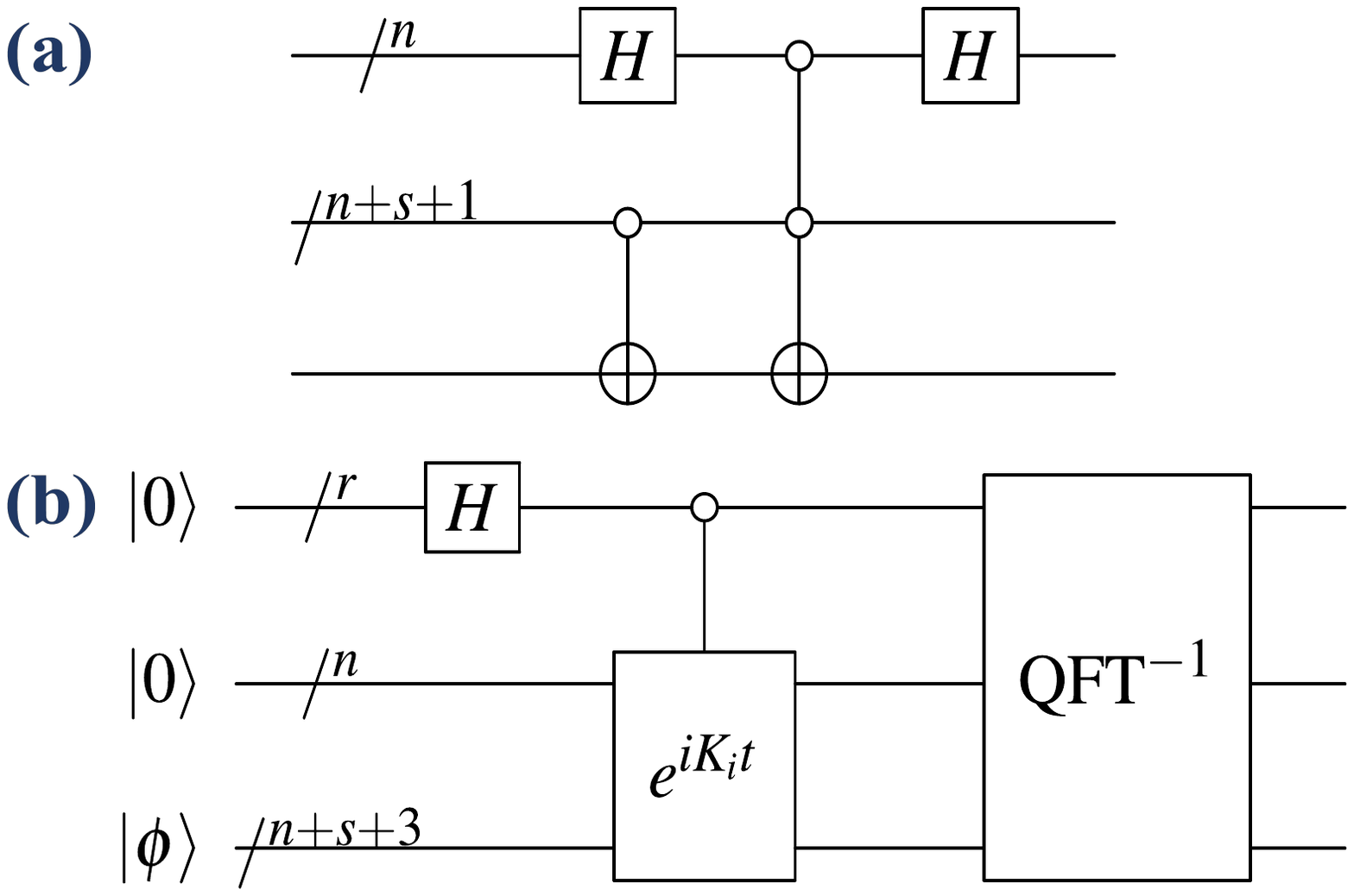}
    \captionsetup[subfloat]{labelsep=none, format=plain, labelformat=empty}
    \subfloat[\label{fig72a}]{}
    \subfloat[\label{fig72b}]{}
    \caption{\label{fig72}(a) Schematic of the circuit of $\mbox{C}_{\Pi}$NOT. (b) Schematic of the eigenvalue estimation circuit $U_i$. The controlled $e^{iK_it}$ represents the block encoded $e^{iK_it}$ with $t$ determined by $r$ control bits. The block named $\mbox{QFT}^{-1}$ is the inverse quantum Fourier transformation circuit. Measuring the first $r$ qubits yields the estimated phases.}
\end{figure}

\begin{figure*}[hbt]
    \centering
    \includegraphics[height=8.5cm,width=12cm]{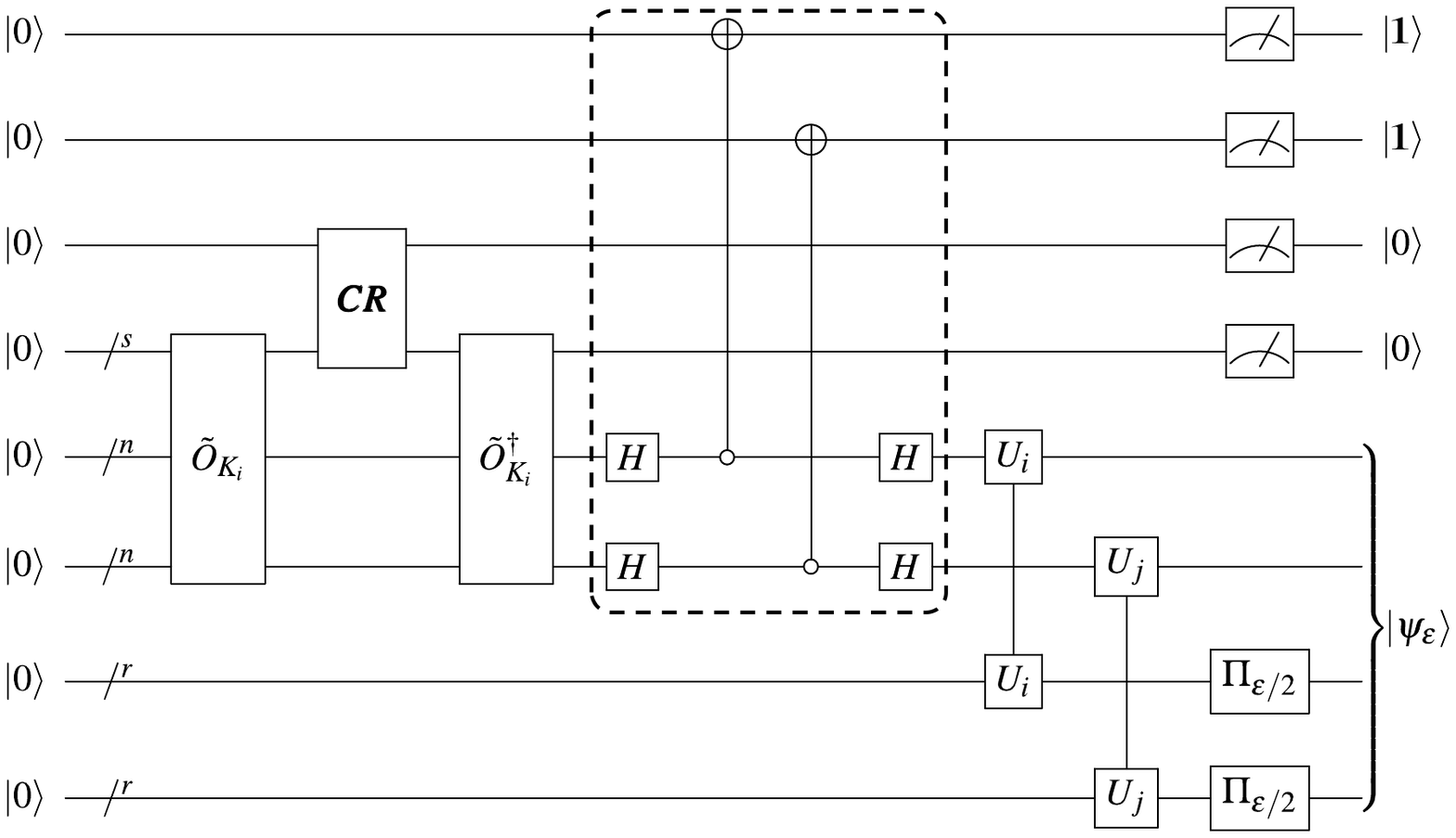}
    \caption{\label{fig2}Schematic of the circuit that prepares state $|\psi_{\epsilon}\rangle$, where $\tilde{O}_{K_i}$ is the block encoding of $K_i$ and $U_i$ is the eigenvalue estimation circuit. Gates enclosed by the dashed line are used for an amplitude encoding of $K_i$ from one of $\tilde{K}_i$.}
\end{figure*}

Now we give details of each subprocedure. Step 1 is to use quantum operations following the route below:
\begin{equation}
    \begin{split}
        |j,k,0,0,0\rangle&\to|j,k,z_{ij}, z_{ik},0\rangle\\
        &\to|j,k,z_{ij}, z_{ik},K(z_{ij}, z_{ik})\rangle\\
        &\to|j,k,0,0,K(z_{ij}, z_{ik})\rangle
    \end{split}
\end{equation}
where the 3 operations are respectively queries to $O_{z_i}$, quantum arithmetics and the inverse $O_{z_i}^{\dagger}$. This is formally stated in Lemma~\ref{lem4}, while proofs of lemmas in this section are left to Appendix.~\ref{app3}.
\begin{mylem}
    \label{lem4}
    The quantum circuit, $\tilde{O}_{K_i}:|j\rangle|k\rangle|z\rangle\to|j\rangle|k\rangle|z\oplus K(z_{ij},z_{ik})\rangle$, can be constructed using $O(\mbox{poly}\log(N/\epsilon))$ quantum gates and $2$ queries to both $O_{z_i}$ and $O^{\dagger}_{z_i}$.
\end{mylem}


The block-encoding of $K_i(Z)$ introduced in the second subprocedure is given in Lemma~\ref{lem4.5}.
\begin{mylem}
    \label{lem4.5}
    An $(N,n+s+1,\epsilon)$-block-encoding of $\tilde{K}_i(Z)$, $\tilde{U}_{K_i}$ is given by
    \begin{equation}
        \begin{split}
        \tilde{U}_{K_i}=&(\textbf{SWAP}_n\otimes I_{s+1})(I_n\otimes H^{\otimes n}\otimes I_{s+1})(\tilde{O}_{K_i}^{\dagger}\otimes I_{1})\\
        &(I_{2n}\otimes\textbf{CR})(\tilde{O}_{K_i}\otimes I_{1})(I_n\otimes H^{\otimes n}\otimes I_{s+1}),
        \end{split}
    \end{equation}
    where $I_n$ and $\textbf{SWAP}_n$ respectively represent the $n$-qubit identity operator and swap gate, and we define $\textbf{CR}:|a\rangle|0\rangle\to|a\rangle(a|0\rangle+\sqrt{1-a^2}|1\rangle)$. An $(N,n+s+2,\epsilon)$-block-encoding of $K_i(Z)$ is then given by $U_{K_i}=\mbox{C}_{\Pi}\mbox{NOT}\cdot \tilde{U}_{K_i}\cdot \mbox{C}_{\Pi}\mbox{NOT}$, using $O(\mbox{poly}\log(N/\epsilon))$ quantum gates and one query to $\tilde{O}_{K_i}$ and its inverse. We define
    \begin{equation}
        \Pi=(I_n-|\vec{1}\rangle\langle\vec{1}|)\otimes|0^{\otimes n+s+1}\rangle\langle0^{\otimes n+s+1}|,
    \end{equation}
    for the uniform superposition state denoted by $|\vec{1}\rangle$, and
    \begin{equation}
        \mbox{C}_{\Pi}\mbox{NOT}=\Pi\otimes X+(I-\Pi)\otimes I.
    \end{equation}
    They give
    \begin{equation}
        \langle j,0^{\otimes n+s+1},1|U_{K_i}|k,0^{\otimes n+s+1},1\rangle=\frac{(K_i)_{jk}}{N}.
    \end{equation}

\end{mylem}
The block-encoding circuit is illustrated in Fig.~\ref{fig71}, with $\mbox{C}_{\Pi}\mbox{NOT}$ shown in Fig.~\ref{fig72a}. The Quantum singular value transformation (QSVT) technique \cite{QSVT2019} then allows the block-encoding of the evolution operator $e^{iK_it/N}$. A phase estimation of $e^{iK_it/N}$ gives the eigenvalue estimation circuit $U_i$ introduced at the beginning of this section, which is illustrated in Fig.~\ref{fig72b} and analyzed in Lemma~\ref{lem4.75}.
\begin{mylem}
    \label{lem4.75}
    Given $0<\epsilon<1$, the eigenvalue estimation circuit $U_i$:
    \begin{equation}
        |u_{ik}\rangle|0\rangle\to|u_{ik}\rangle|\frac{\tilde{\lambda}_{ik}}{N}\rangle,
    \end{equation}
    can be implemented using $\tilde{O}(\frac{1}{\epsilon})$ queries to $U_{K_i}$ and elementary gates, where $|\lambda_{ik}-\tilde{\lambda}_{ik}|\leq N\epsilon$.
\end{mylem}

The circuit of the third subprocedure is illustrated in Fig.~\ref{fig2}, where two eigenvalue estimation circuits are implemented. Projections $\Pi_{\epsilon/2}$ are included to eliminate terms corresponding to $\tilde{\lambda}_{ik}/N<\epsilon/2$. The output is denoted by $|\psi_{\epsilon}\rangle$, as stated formally in Lemma~\ref{lem5}.
\begin{mylem}
    \label{lem5}
    Given $\epsilon>0$, the preparation of either:
    \begin{equation}
        |\phi_{\epsilon}\rangle=\sum_{k=1}^{M_i}|u_{ik}\rangle|\frac{\tilde{\lambda}_{ik}}{N}\rangle,
        \label{eq28.9}
    \end{equation}
    or 
    \begin{equation}
        |\psi_{\epsilon}\rangle=\sum_{k=1}^{M_i}\sum_{l=1}^{M_j}\frac{\lambda_{ik}}{N}\langle \vec{u}_{jl},\vec{u}_{ik}\rangle|u_{ik}\rangle|u_{jl}\rangle|\frac{\tilde{\lambda}_{ik}}{N}\rangle|\frac{\tilde{\lambda}_{jl}}{N}\rangle,
        \label{eq29}
    \end{equation}
    can be realized within $\tilde{O}(\frac{1}{\epsilon})$ elementary gates and one query to $U_{i}$ and $U_{j}$, where $|\tilde{\lambda}_{ik}-\lambda_{ik}|<N\epsilon$. $M_i,M_j=O(\log N)$ is the number of kept eigenvalues of $K_i, K_j$ after the projection $\Pi_{\epsilon/2}$. 
\end{mylem}

Denote by $O_{\epsilon}$ the quantum circuit preparing $|\psi_{\epsilon}\rangle$, with which repetitive executions of phase estimation or amplitude estimation yields the value of $\lambda_{ik}$ and $\lambda_{ik}\langle|\vec{u}_{jl},\vec{u}_{ik}\rangle|$. Thus $R_{\kappa}$ and its determinant are obtained according to Eq.~\ref{eq13}. Selecting carefully the estimation precision and adding the preprocessing subprocedure for replacement of $\bm{z}$ by $\bm{Wy}$ give the overall process estimating the contrast function with the complexity of $\tilde{O}(1/\epsilon^2_1)$, as we stated in Theorem.~\ref{the1} and~\ref{the2}. Their proofs are also in Appendix~\ref{app3}.

\subsection{\label{sec3.4}Optimization}


The minimization of the function with its argument being an orthogonal matrix is done by techniques for optimization on a Stiefel manifold as in \cite{OptiStieMani1998,KICA2002}. The remaining questions are how to utilize the difference between two theorems and whether the optimization runs successfully. 

Theorem~\ref{the1} allows us to first speculate the global property of $J(\bm{W}Y)$ when varying $\bm{W}$, and once $J(\bm{W}Y)$ performs like that $\bm{W}\bm{y}$ is near-independent, we can try higher accuracy without introducing much more complexity using Theorem~\ref{the2} in view of the $\epsilon_2^{-1}$ dependence. In the near-independent case, as analyzed in Appendix~\ref{app4}, $R_{\kappa}$ is an identity with perturbation of magnitude $\epsilon_2$ so that $J(\bm{W}Y)\approx1-\zeta\epsilon_2$ where $\zeta$ is some constant dependent on $F$. Therefore the $\epsilon_1\epsilon_2$ precision can always distinguish $\epsilon_2=a$ or $\epsilon_2=\frac{\zeta+\epsilon_1}{\zeta-\epsilon_1}a$ which allows the optimization to proceed.

Regarding the influence of the input error induced by preprocessing, we have
\begin{equation}
    \begin{split}
    \tilde{M}^{-1/2}&=(I+(\tilde{M}^{-1/2}-M^{-1/2})M^{1/2})M^{-1/2}\\
    &=(I+\epsilon_2E\sqrt{\|M\|_2})M^{-1/2}+o(\epsilon_2^2)
    \end{split}
\end{equation}
for some $E$ with $\|E\|_2<1$ in the limit as $\epsilon_2$ approaches 0. Using a first-order approximation with respect to $\epsilon_2$, $\tilde{M}^{-1/2}$ can be viewed as the product of matrix $M^{-1/2}$ and an orthogonal $U_E=I+\epsilon_2E\sqrt{\|M\|_2}$ since $M$ is symmetric and $E$ as well. Therefore, the optimization result only differs from the theoretically optimal solution by an orthogonal transformation, namely $\bm{W}=\tilde{\bm{W}}U_E^{-1}$. Otherwise, large $\epsilon_2$ may lead to that $\tilde{M}^{-1/2}$ is not orthogonal at all and meaningless results.

The optimization is further restricted by two factors when $\epsilon_2$ approaches $1/\sqrt{N}$. One is that the statistical error plays an important role since the samples are in general not perfect and the other is the higher demand for the accuracy for distinguishability.


\section{\label{sec4.0}Numerical experiments}
We present numerical tests for assessment of the capability of our algorithm to separate the ICs and to support Theorem~\ref{the1}. In the first half, the raw data is sampled by the given PDF, as illustrated in Fig.~\ref{fig3}, and a biological dataset is used later to introduce examples of applications.

\begin{figure}[ht]
    \centering

    \includegraphics[height=2.8cm,width=8cm]{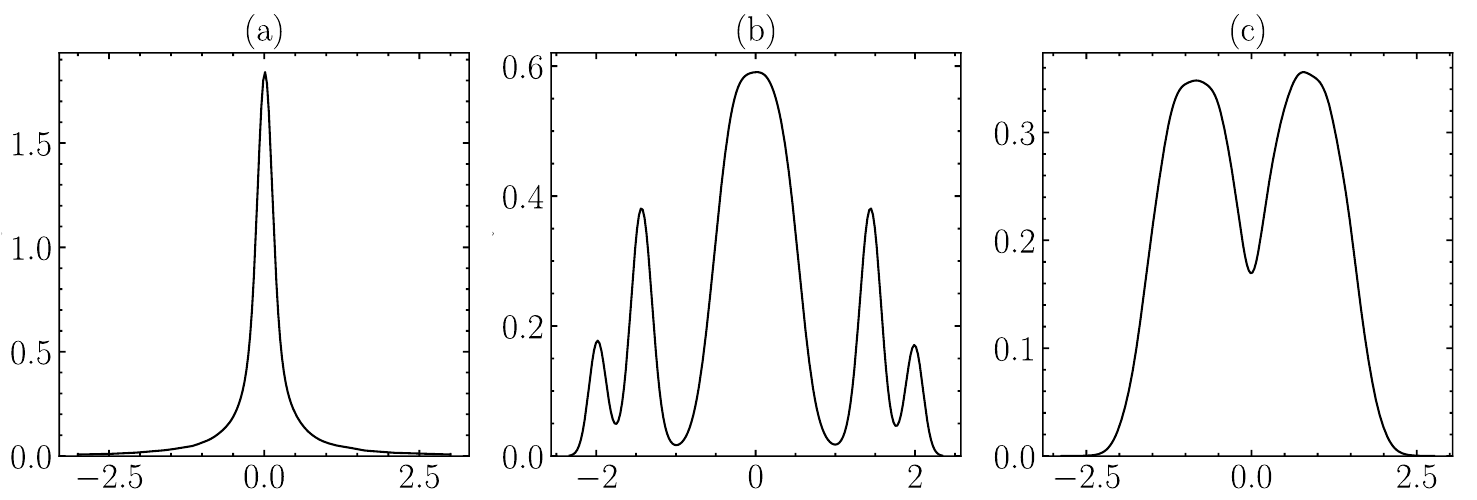}
    \captionsetup[subfloat]{labelsep=none, format=plain, labelformat=empty}
    \subfloat[\label{fig3a}]{}
    \subfloat[\label{fig3b}]{}
    \subfloat[\label{fig3c}]{}
    \caption{\label{fig3}Probability density functions of partial ICs.}
\end{figure}

In the first stage, data sampled from a given PDF are seen as samples of an IC $s_i$. After a linear transformation $\bm{A}$, mixtures $\bm{A}S$ are sent for computing the contrast function so that it is convenient to observe its performance as $\bm{A}$ varies. Fig.~\ref{fig4a} reveals the performance of the contrast function used in the original KICA algorithm, while Fig.~\ref{fig4b},~\ref{fig4c} and~\ref{fig4d} illustrate cases using the adapted one with measuring error, $\epsilon_1=0, 2, 4\times10^{-3}$, respectively. Contrast functions can always approach their minima in our test. We find the adaptation of the contrast function does not appear to destroy the effectiveness after comparing these results.

\begin{figure}[ht]
    \centering
    \includegraphics[height=6.9cm,width=8.5cm]{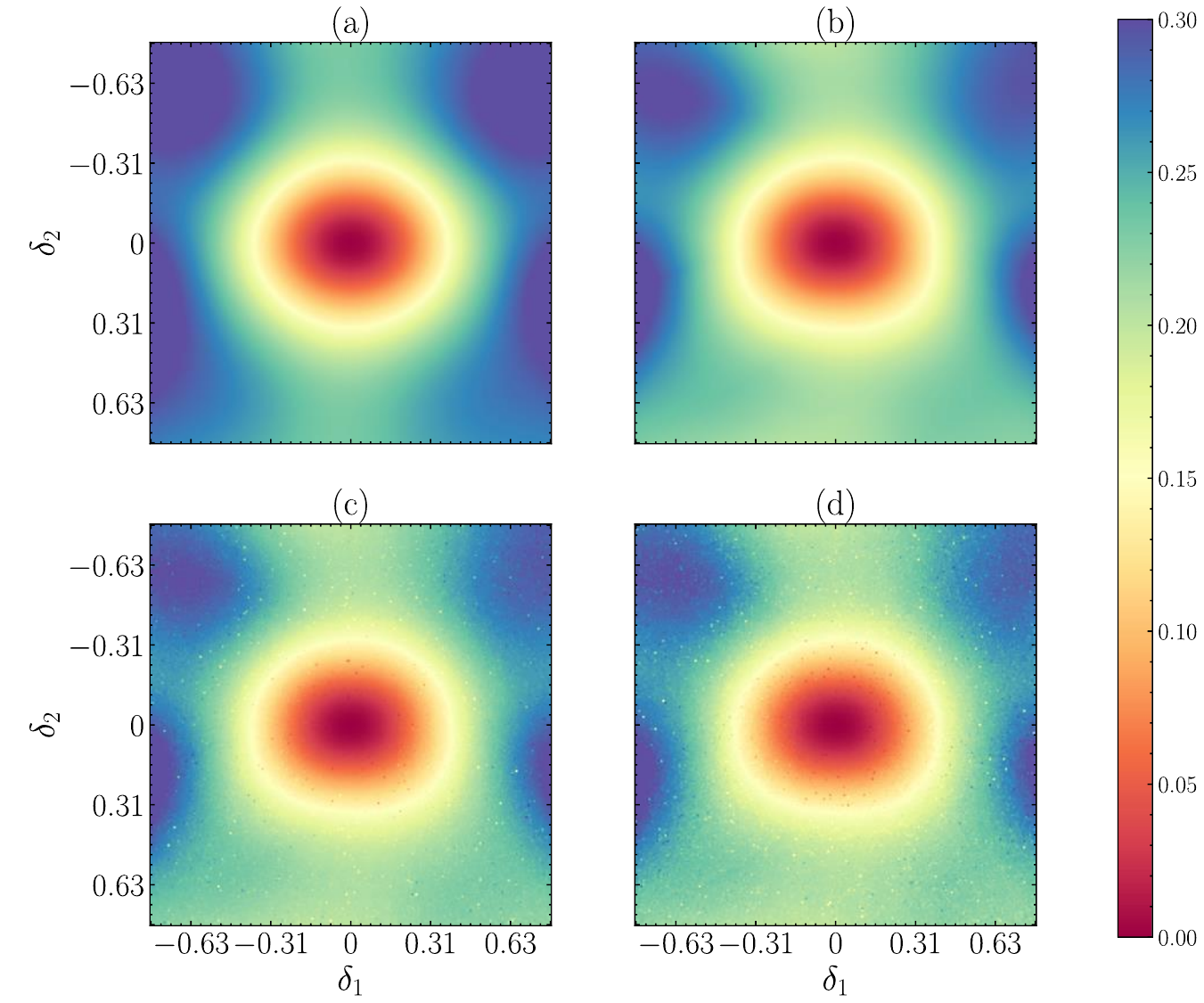}
    \captionsetup[subfloat]{labelsep=none, format=plain, labelformat=empty}
    \subfloat[\label{fig4a}]{}
    \subfloat[\label{fig4b}]{}
    \subfloat[\label{fig4c}]{}
    \subfloat[\label{fig4d}]{}
    \caption{\label{fig4}Values of $-\ln(\det(R_{\kappa}))$ represented by colors when varying the mixing coefficients, $\exp\{i(\delta_1P_1+\delta_2P_2)\}$. $P_1, P_2$ are two generators, and $\delta_1,\delta_2$ are parameters reflecting the deviation of the mixtures from the original ICs. (a) $-\ln(\det(R_{\kappa}))$ computed with the correct signs of $\langle\vec{u}_{ik},\vec{u}_{jl}\rangle$; (b), (c), (d) $-\ln\det(R_{\kappa})$ computed with measuring error $\epsilon_1=0, 2, 4\times10^{-3}$ respectively, signs of entries of $R_{\kappa}$ neglected.}
\end{figure}

Vary sample size and mixing coefficients at the same time, then we compute norms of states of $|K_i\rangle$ and $|\psi_{\epsilon}\rangle$, results given in Fig.~\ref{fig5a} and ~\ref{fig5b}. This reveals that only the extent of mixing is dominant for the success rate problem, and the probability problem occurs mainly when the mixture tends to be unmixed.

The relation between $\big\||\psi_{\epsilon}\rangle\big\|$ and sample size $N$ is further shown in Fig.~\ref{fig6a}. We set mixing coefficients to be $\exp\{i\delta P\}$, a fixed $\delta$ for each curve. It is found that in double logarithmic coordinates, the slope of each curve increases gradually from 1/2 to 1 as $N$ goes up except in the case that $\delta=0$, where the slope stays at 1/2. 
Our fitted curves agree with our expectation that $\big\||\psi_{\epsilon}\rangle\big\|=O(\epsilon_2+1/\sqrt{N})$, while in this case $\epsilon_2$ can be identified with $\delta$. 

Fix the mixing coefficients to $\exp\{i(0.1\pi P)\}$, we illustrate the relation between the error of $\det(R_{\kappa})$ and $\epsilon_1$ as $N$ increases in Fig.~\ref{fig6b}. Because of the significant effect of randomness, the fit has a certain variance but it is enough to show the error increases linearly with $\epsilon_1$ but hardly with $N$. 
\begin{figure}[ht]
    \centering
    \captionsetup[subfloat]{labelsep=none, format=plain, labelformat=empty}
    \subfloat[\label{fig5a}]{}
    \subfloat[\label{fig5b}]{}
    \includegraphics[height=4cm,width=8cm]{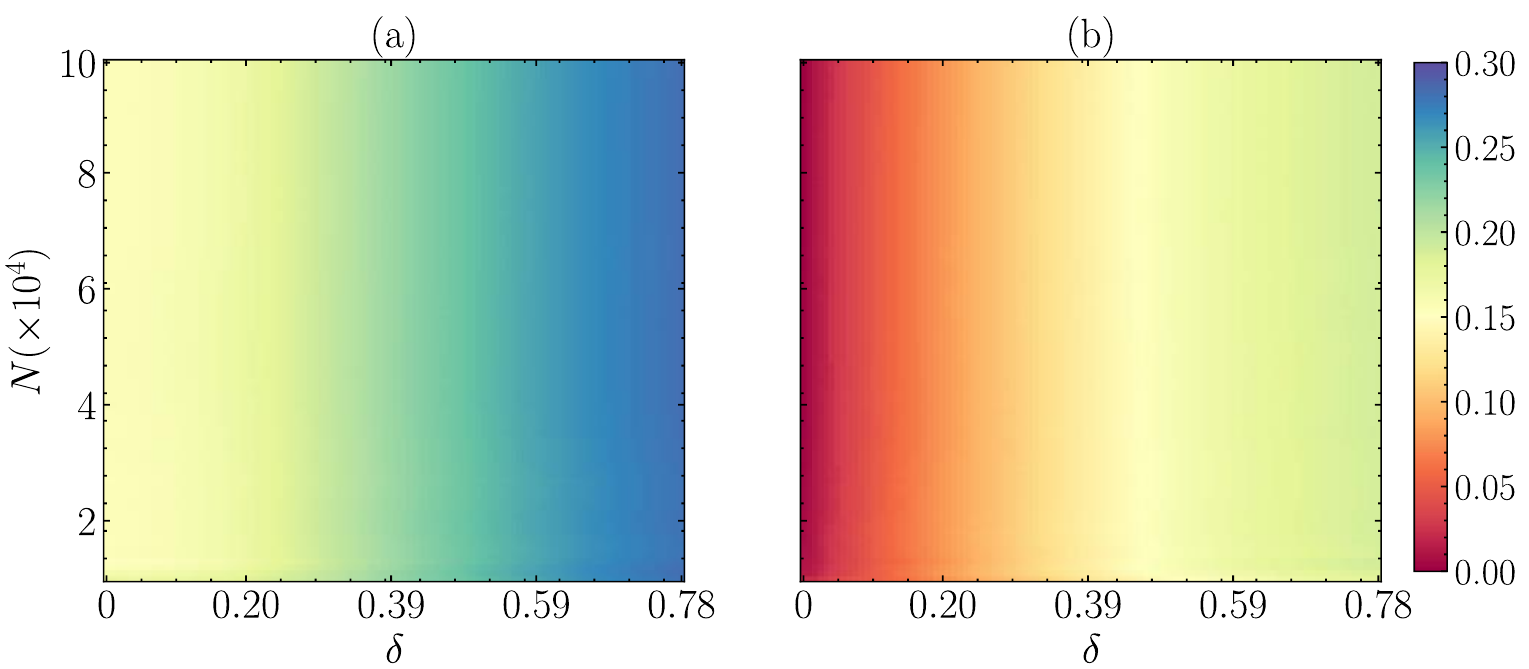}
    \caption{\label{fig5} Represented by colors, norms of different quantum states computed with mixing coefficients $\exp\{i\delta P\}$ as $N$ and $\delta$ increase. (a) Norms of $|K_i\rangle$; (b) Norms of $|\psi_{\epsilon}\rangle$.}
\end{figure}

\begin{figure}[ht]
    \centering
    \subfloat[\label{fig6a}]{
        \includegraphics[height=3.6cm,width=4cm]{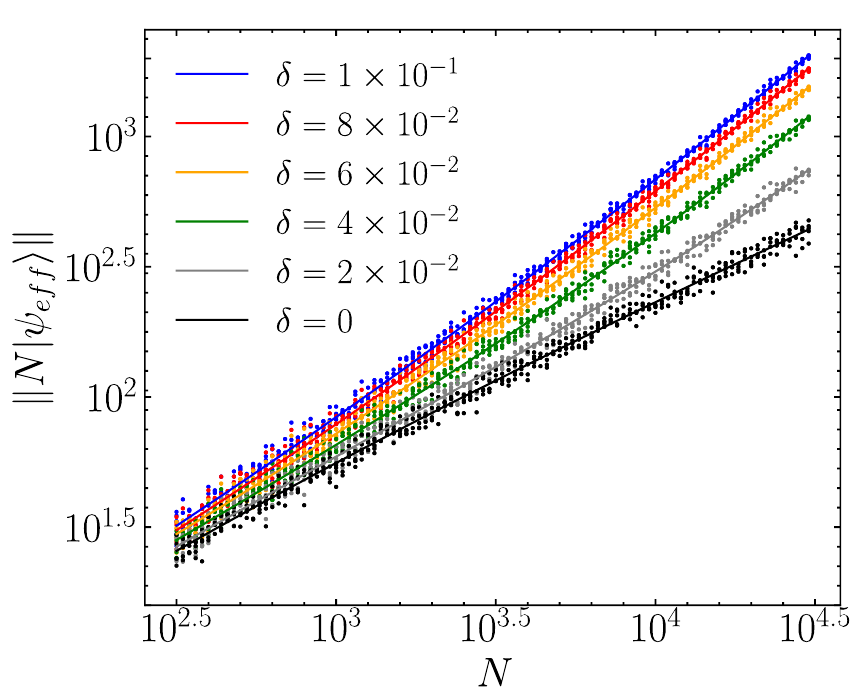}
    }
    \subfloat[\label{fig6b}]{
        \includegraphics[height=3.5cm,width=4cm]{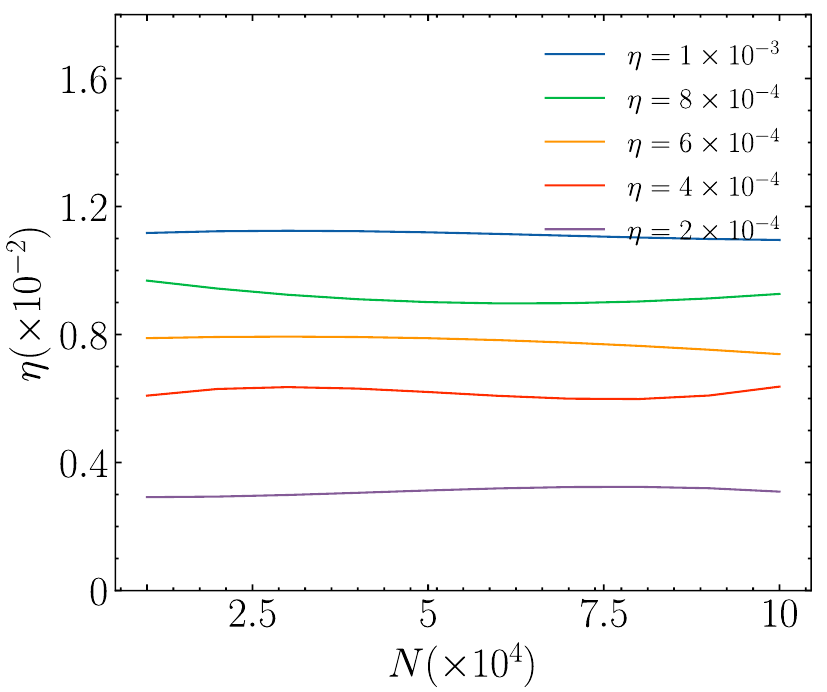}
    }
    
    \caption{\label{fig6} (a) The norm of $|\psi_{\epsilon}\rangle$ computed with mixing coefficients $\exp\{i\delta P\}$ as sample size $N$ and $\delta$ vary and $P$ is one of the generators. (b) Relative error of $\det(R_{\kappa})$, $\eta=|(\det(R_{\kappa})-\det(\tilde{R}_{\kappa}))/\det(R_{\kappa})|$, as a function of $N$ when introducing measuring error to be $\epsilon$.}
\end{figure}

\begin{figure}[ht]
    \centering
    \subfloat[\label{fig7a}]{
        \includegraphics[height=3.8cm,width=4.3cm]{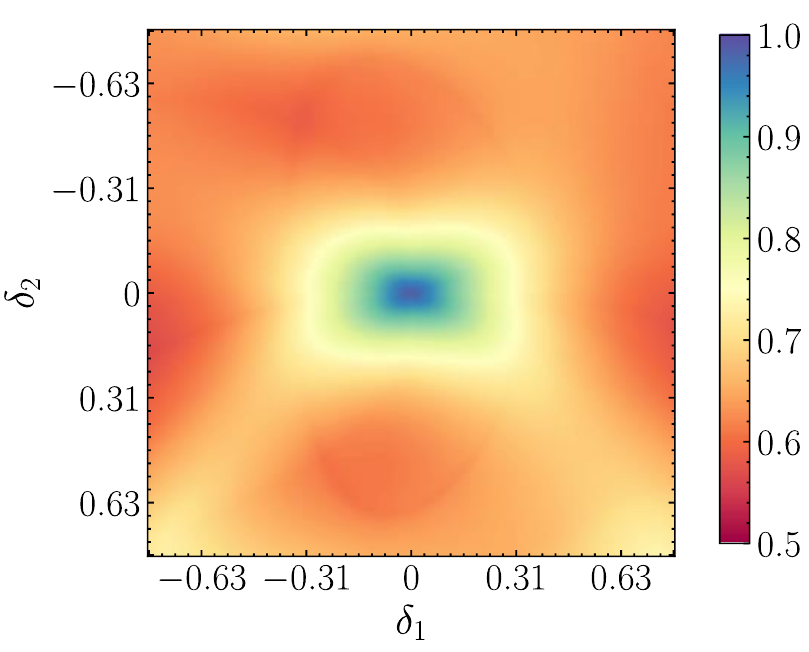}
    }
    \subfloat[\label{fig7b}]{
        \includegraphics[height=3.8cm,width=4cm]{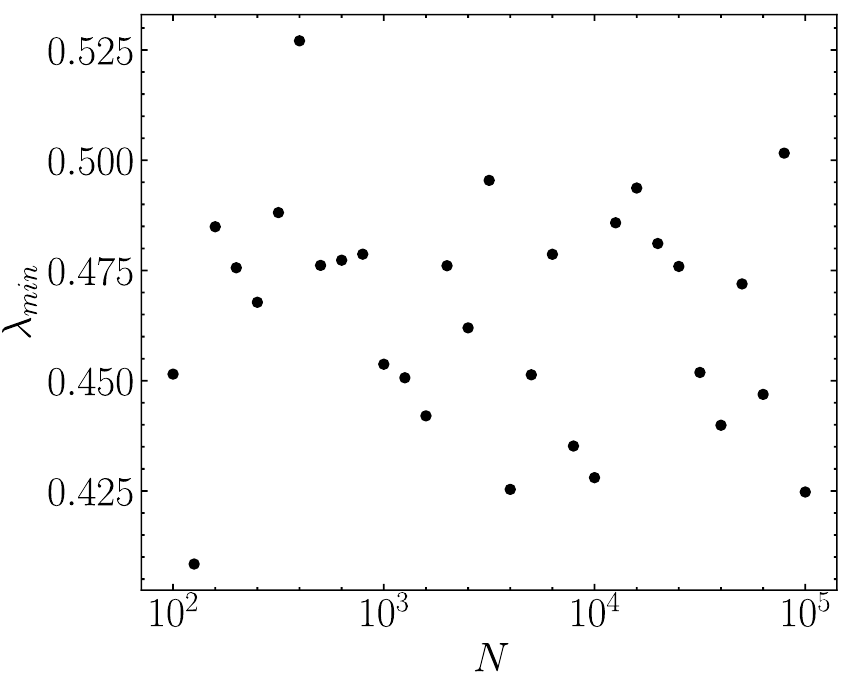}
    }
    \caption{\label{fig7} (a) The minimal eigenvalues of $R_{\kappa}$ for various mixtures, $\exp\{i(\delta_1P_1+\delta_2P_2)\}S$, where $S$ are given ICs. (b) For fixed sample size, each datapoint represents the lowest one among all the minimal eigenvalues of $R_{\kappa}(\bm{W}S)$ for randomly selected mixing coefficients $\bm{W}$.}
\end{figure}

\begin{figure}[ht]
    \centering
    \subfloat[\label{fig8a}]{
        \includegraphics[height=3.3cm,width=3.6cm]{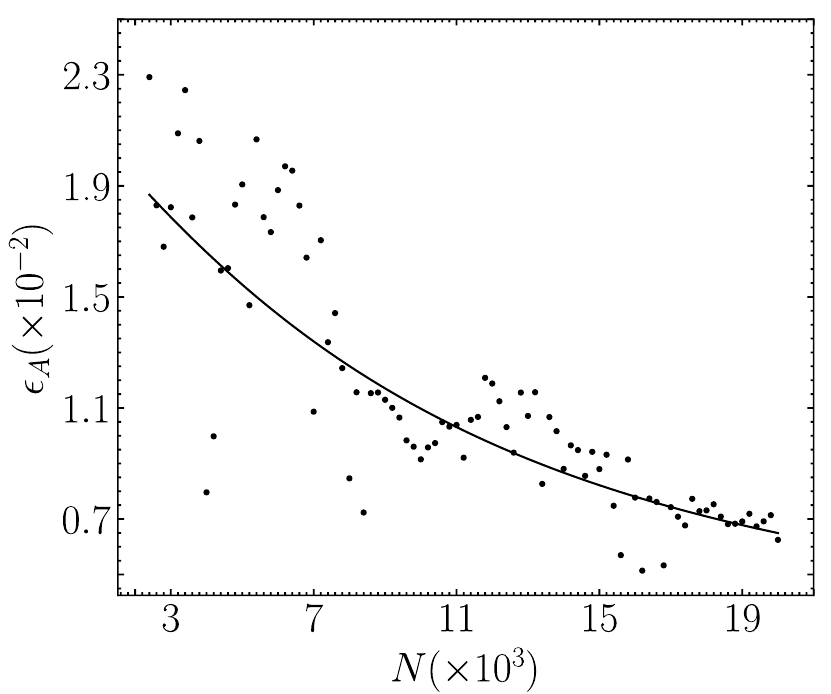}
    }
    \subfloat[\label{fig8b}]{
        \includegraphics[height=3.5cm,width=4cm]{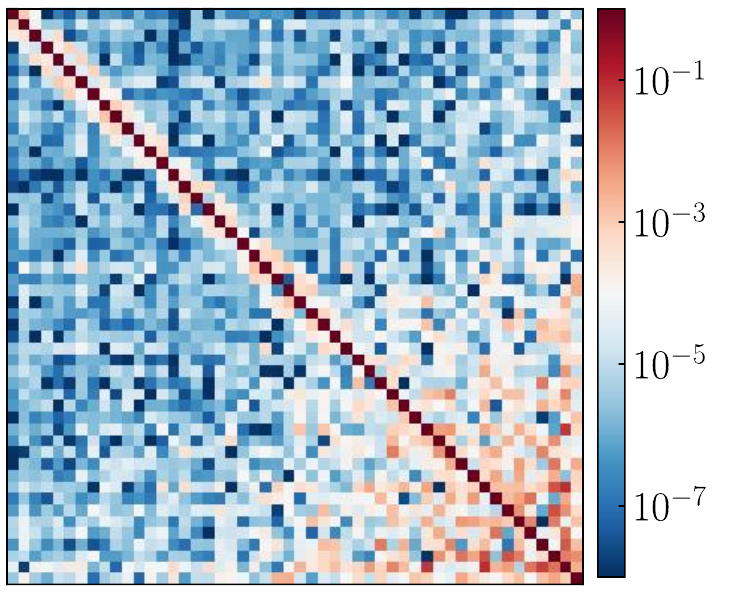}
    }
    \caption{\label{fig8} (a) Amari error $\epsilon_A$ between the optimized mixing coefficients and the pre-known one as $N$ increases. (b) Correlation matrix between ICs computed by classically simulating our algorithm and by the Icasso algorithm called from the stabilized-ica package. The color of the grid in row $i$, column $j$ corresponds to the correlation coefficient between the $i$-th and the $j$-th ICs separated with each method.}
\end{figure}

With the same setup where $\delta_1$ and $\delta_2$ are variable, the minimum eigenvalues $\xi$ of $R_{\kappa}$ are given in Fig.~\ref{fig7a}. We also repeat randomly choosing mixing coefficients for $\xi$ and select the minimal for each fixed sample size. Each data point in Fig.~\ref{fig7b} represents the lowest $\xi$ we have met during our test for corresponding $N$, where no apparent numerical relationship appears even in such large ranges. Our upper bound of the error of $J(X)$ seems not to fail in more general cases.

Finally, we classically simulate the optimizer together with our contrast function computation subprogram. With the ICs shown in Fig.~\ref{fig3}, we recorded the optimized unmixing coefficient matrices which are then assessed by a metric named "Amari error" \cite{Amari1995, KICA2002}. If the known mixing coefficients matrix is $\bm{A}$, while the optimized unmixing matrix reads $\bm{W}^{-1}$, the Amari error is defined by
\begin{equation}
    \epsilon_A(\bm{A}, \bm{W}) = \frac{1}{2m}\sum_{i=1}^m(\frac{\sum_{j=1}^m|a_{ij}|}{\max_j|a_{ij}|}-1)+\frac{1}{2m}\sum_{j=1}^m(\frac{\sum_{i=1}^m|a_{ij}|}{\max_i|a_{ij}|}-1),
\end{equation}
where $a_{ij}=(\bm{A}\bm{W}^{-1})_{ij}$. Fig.~\ref{fig8} shows that $\epsilon_A$ decreases with $N$, which is intuitive for a feasible ICA algorithm.

To illustrate an example of applications, we introduce a transcriptomic dataset \cite{BioData2017}. It contains the expression of genes for patients with lung cancer. We simulate our algorithm to separate the latent ICs from this dataset with a simplified optimization process. Compare our results of decomposition with the one obtained by stabilized-ica (sica) package based on Icasso algorithm \cite{Icas2003}, available from \href{https://github.com/ncaptier/stabilized-ica}{https://github.com/ncaptier/stabilized-ica}. The correlation matrix between ICs separated in different ways is shown in Fig.~\ref{fig8b}. A clear correlation exists between two sets of ICs.

\section{\label{sec5.0}Conclusion}
In this paper, we have presented a quantum ICA algorithm scaling as $\tilde{O}(1/\epsilon_1^2)$ to estimate a contrast function adapted from the one of classical KICA algorithm with the precision of $\epsilon_1$. Our method exponentially accelerates the computation of KICA. This speedup comes from the better performance of matrix operations offered by quantum arithmetics, QSVT, amplitude estimation, etc. We have provided an upper bound of the error of the estimated $R_{\kappa}$ which gives the scaling of the error of the contrast function. This verifies that the exponential speedup is generally achievable. Compared to the existing quantum algorithms for ICA based upon kurtosis or entropy, our method seems to be a compromise between robustness and efficiency.

In Appendix~\ref{app4}, we asymptotically analyze why the Gaussian kernel is an appropriate choice for the KICA algorithm by research of $R_{\kappa}(Z)$, whose off-diagonal elements cannot be all zero unless all the random variables $z_i$ are mutually independent if Gaussian kernel is chosen and $z_i$ are near-independent. It will occur that the contrast function decreases from the center to the periphery when $Z$ is taken among different perturbations of independent components regardless of the statistical error.

We implemented numerical experiments with raw data sampled from given distributions and a biomedical dataset. These samples were then sent for estimation of contrast function, success rate, error performance, the effectiveness of separation, etc. to ensure the reliability and further reveal the performance of our adapted algorithm. The error and the success rate were also checked to be consistent with our analysis.

However, we have not yet considered whether such distributions exist that $R_{\kappa}(Z)=I$ but $\bm{z}_i$ are not mutually independent at all. If they exist, the performance around this extremum is left for further consideration.

It is seen that we perform a principal components analysis of $K_i^TK_i$ which may be worse than directly analyzing $K_i$ which might be realized by the block-encoding technique introduced in \cite{Gramblock2022}. Moreover, perhaps the input model can be replaced by $O_{f_{x_i}}:|0\rangle\to\sum_x\sqrt{f_{x_i}(x)}|x\rangle$, for PDF of $x_i$, $f_{x_i}$. Such input model gives the possibility to estimate statistical properties with the precision of $\epsilon$ using $O(\epsilon^{-1})$ queries \cite{EntropyEst2021,Statpro2022}, compared to $O(\epsilon^{-2})$ sampling complexity in classical cases. It is also possible that the model similarly reduces the sampling complexity when computing the contrast function.

\acknowledgments
This work was supported by the National Natural Science Foundation of China (Grant No. 12034018) and Innovation Program for Quantum Science and Technology (Grant No. 2021ZD0302300).

The numerical calculations in this paper have been done on the supercomputing system in the Supercomputing Center of University of Science and Technology of China.
\appendix

\section{\label{app1}ICA based on statistics and quantum speedup}
Non-Gaussianity has been widely used as a measure of the independence in some ICA models. Such as kurtosis or mutual information, statistics can naturally represent the non-Gaussianity \cite{ICA2002}. As introduced later, existing quantum algorithms have shown advantages in estimating these statistics.


For simplicity, we assume components $\bm{y}_i$ of the random vector $\bm{y}$ have been whitened. And quantum oracles $O_{y_i}$ are given. 
We introduce the two measures of non-Gaussianity and their estimation on a quantum computer in this section. The corresponding optimization is similar to what we have introduced in our algorithm.

\subsection{\label{app1.1}ICA based upon kurtosis}
The kurtosis of $\bm{y}_i$ is defined by
\begin{equation}
    \mbox{kurt}(\bm{y}_i)=E\{\bm{y}_i^4\}-3(E\{\bm{y}_i^2\})^2.
\end{equation}
Suppose each IC $\bm{s}_i$ has kurtosis with the same sign. It is proved that the mixture, $\bm{a}^T\bm{y}$ where $\|\bm{a}\|_2=1$, has the maximal absolute kurtosis when $\bm{a}^T\bm{y}$ is properly one of the ICs\cite{ProfKurt1995}. Therefore, optimizing $|\mbox{kurt}(\bm{a}^T\bm{y})|$ gives one of the ICs. 

The quantum mean estimator introduced can achieve the estimation, just by estimating using the oracle $O_{y_i^k}:|j\rangle|0\rangle\to|j\rangle|y_{ij}^k\rangle$ constructed by quantum arithmetics. The query and gate complexity is $\tilde{O}(1/\epsilon)$, for $\epsilon$ being the estimation precision. It is worth noting that this methodology can be applied to other ICA algorithms using mean estimators such as the notable FastICA algorithm \cite{FastICA1999}.

\subsection{\label{app1.2}ICA based upon mutual information}
Mutual information is a strong candidate as a choice of contrast function, which can be defined by Shannon's entropy. 
The Shannon entropy of $\bm{y}$ is 
\begin{equation}
    H(\bm{y})=-\int f_{y}(\bm{\alpha})\log f_{y}(\bm{\alpha})d\bm{\alpha},
\end{equation}
where $f_{y}$ is the PDF of $\bm{y}$. It reads
\begin{equation}
    H(\bm{y})=-\sum_{\bm{\gamma}}p_{\bm{\gamma}}\log p_{\bm{\gamma}},
\end{equation}
for a discrete one, where $p_{\bm{\gamma}}$ is the probability that $\bm{y}_i=\gamma_i,i=1,2,\dots,m$. Then the mutual information $I$ between $\bm{y}_i$ is given by
\begin{equation}
    I(\bm{y})\equiv I(\bm{y}_1,\bm{y}_2,\dots,\bm{y}_m)=\sum_{i=1}^{m}H(\bm{y}_i)-H(\bm{y}).
\end{equation}
It is nonnegative, and zero if and only if $\bm{y}$ are mutually independent, and therefore a natural measure of the independence. 

A quantum Shannon's entropy estimator exists, as introduced in \cite{Entroest2018}, with both query and gate complexity of $\tilde{O}(\frac{\sqrt{N}}{\epsilon^2})$. 
For comparison, $\tilde{O}(1/\epsilon)$ queries suffices to estimate $E\{\bm{y}_i^p\}$ or $\mbox{kurt}(\bm{y}_i)$ and $\tilde{O}(1/\epsilon^2)$ queries are required for the adapted KICA contrast function. The kurtosis is known to be not robust so our adapted one may serve as a compromise between efficiency and robustness.



\section{\label{app2}Error analysis}
We refer to $\tilde{a}$ as an estimate of some variable $a$ in this section.
\subsection{\label{app2.1}The error of preprocessing}
We first analyze the error of the constructed $O_{y'_i}$ in Eq.~\ref{eq16real}. 
We denote the covariance matrix of $\bm{x}$ by $M$, and we have
\begin{equation}
    \begin{split}
        |\tilde{M}_{ij}-&M_{ij}|\leq|\tilde{E}\{\bm{x}_i\bm{x}_j\}-\tilde{E}\{\bm{x}_i\bm{x}_j\}|+\\
        &|\tilde{E}\{\bm{x}_i\}\tilde{E}\{\bm{x}_j\}-E\{\bm{x}_i\}E\{\bm{x}_j\}|<\frac{1}{2}\frac{\|M\|_2^{3/2}\epsilon}{\mu_M^2m}\\ 
    \end{split}
\end{equation}
Next, we estimate the error of the inverse of the covariance matrix $M^{-1}$, where we denote $D\equiv\tilde{M}-M$:
\begin{equation}
    \frac{\|M^{-1}-M^{-1}\|_2}{\|M^{-1}\|_2}\leq\frac{\|M^{-1}D\|_2}{1-\|M^{-1}D\|_2}.
    \label{eqap1}
\end{equation}
Eq.~\ref{eqap1} works only if $\|M^{-1}D\|_2\leq1$. Note that 
\begin{equation}
    \|M^{-1}D\|_2\leq\frac{\mu_M}{\|M\|_2}\times m\times\frac{1}{2}\epsilon\frac{\|M\|_2^{3/2}}{\mu_M^2m}=\frac{1}{2}\frac{\epsilon\sqrt{\|M\|_2}}{\mu_M}<1.
\end{equation}
The first inequality comes from that $\|D\|_2\leq\|D\|_F\leq m\times\max_{i,j}\{D_{ij}\}$. Now we may apply the inequality and have,   
\begin{equation}
    \|\tilde{M}^{-1}-M^{-1}\|_2\leq\frac{\frac{1}{2}\|M\|_2^{-1/2}\epsilon}{1-\frac{1}{2}\frac{\epsilon\sqrt{\|M\|_2}}{\mu_M}}.
\end{equation}
Introduce a theorem for further estimation.
\begin{mythe}\cite{matanaly2013}
    \label{the7}
    For two positive definite matrix $A,B$, if eigenvalues of $A,B$ are bounded below some by positive number $a$, thus we have the inequality
    \begin{equation}
        |||A^r-B^r|||\leq ra^{r-1}|||A-B|||,\quad0<r<1,
    \end{equation}
    where $|||\cdot|||$ refers to a unitarily invariant norm including 2-norm we used later.
\end{mythe}
The maximal eigenvalue of $\tilde{M}$ is bounded above by $\|M\|_2+\|D\|_2<\|M\|_2+\epsilon\frac{\|M\|_2^{3/2}}{2\mu_M^2}$. Replace $A,B,a,r$ in Theorem~\ref{the7} by $M^{-1}, \tilde{M}^{-1}, (\|M\|_2+\epsilon\frac{\|M\|_2^{3/2}}{2\mu_M^2})^{-1}, \frac{1}{2}$, we have:
\begin{equation}
    \begin{split}
        &\|\tilde{M}^{-\frac{1}{2}}-M^{-\frac{1}{2}}\|_2\\
        &<\frac{1}{2}(\|M\|_2+\epsilon\frac{\sqrt{\|M\|_2}}{2\mu_M^2})^{\frac{1}{2}}\frac{\frac{1}{2}\|M\|_2^{-1/2}\epsilon}{1-\frac{1}{2}\frac{\epsilon\sqrt{\|M\|_2}}{\mu_M}}\\
        &<\frac{1}{2}\sqrt{\frac{3}{2}}\sqrt{\|M\|_2}\times\|M\|_2^{-1/2}\epsilon<\epsilon.
    \end{split}
\end{equation}


\subsection{\label{app2.2}The error of entries of \texorpdfstring{$R_{\kappa}$}{}}
In this section, the error bound of $\frac{\lambda_{ik}}{\lambda_{ik}+N\kappa/2}\frac{\lambda_{jl}}{\lambda_{jl}+N\kappa/2}\langle\vec{u}_{ik},\vec{u}_{jl}\rangle$ in cases of Theorem~\ref{the1} and~\ref{the2} is proved.

Use notations that $\epsilon_{\mu}$ and $\epsilon_I$ represent respectively the error of $\mu_{ik}$  induced by phase estimation and of $\mu_{ik}\langle\vec{u}_{ik},\vec{u}_{jl}\rangle$ induced by amplitude estimation for simplicity. Only entries with either $\mu_{ik}$ or $\mu_{jl}<\epsilon_{\mu}$ may be discarded. The discarded entry where $\mu_{ik}<\epsilon_{\mu}$ reads
\begin{equation}
    \frac{\mu_{ik}}{\mu_{ik}+\kappa/2}\frac{\mu_{jl}}{\mu_{jl}+\kappa/2}\langle\vec{u}_{ik},\vec{u}_{jl}\rangle<\frac{\epsilon_{\mu}}{\epsilon_{\mu}+\kappa/2}\times1<\frac{\xi\epsilon_1}{2},
\end{equation}
for the first case, and
\begin{equation}
    \frac{\mu_{ik}}{\mu_{ik}+\kappa/2}\frac{\mu_{jl}}{\mu_{jl}+\kappa/2}\langle\vec{u}_{ik},\vec{u}_{jl}\rangle<\frac{2}{\kappa}\epsilon_{\mu}\times G\epsilon_2=\frac{\epsilon_1\epsilon_2}{2},
\end{equation}
for the second, while this analysis also works when $\mu_{jl}<\epsilon_{\mu}$. For entries kept, the following results can be then computed:
\begin{align}
    &|\frac{\mu_{ik}}{\mu_{ik}+\kappa/2}-\frac{\tilde{\mu}_{ik}}{\tilde{\mu}_{ik}+\kappa/2}|<\frac{\kappa/2}{\tilde{\mu}_{ik}+\kappa/2}\frac{\epsilon_{\mu}}{\mu_{ik}+\kappa/2},\\
    &|\frac{1}{\mu_{ik}+\kappa/2}-\frac{1}{\tilde{\mu}_{ik}+\kappa/2}|<\frac{1}{\tilde{\mu}_{ik}+\kappa/2}\frac{\epsilon_{\mu}}{\mu_{ik}+\kappa/2}.
\end{align}
Combining these two errors and $\epsilon_I$, we give the error bound of entries by:
\begin{equation}
    \begin{split}
        &\frac{\kappa/2}{\tilde{\mu}_{jl}+\kappa/2}\frac{\epsilon_{\mu}}{\mu_{jl}+\kappa/2}\frac{\mu_{ik}\langle\vec{u}_{ik},\vec{u}_{jl}\rangle}{\mu_{ik}+\kappa/2}+\\
        &\frac{\tilde{\mu}_{jl}\mu_{ik}\langle\vec{u}_{ik},\vec{u}_{jl}\rangle}{\tilde{\mu}_{jl}+\kappa/2}\frac{1}{\tilde{\mu}_{ik}+\kappa/2}\frac{\epsilon_{\mu}}{\mu_{ik}+\kappa/2}+\epsilon_I\frac{\tilde{\mu}_{jl}}{\tilde{\mu}_{jl}+\kappa/2}\frac{1}{\tilde{\mu}_{ik}+\kappa/2}\\
        &<\frac{\epsilon_{\mu}\mu_{ik}}{\tilde{\mu}_{jl}+\kappa/2}\frac{\langle\vec{u}_{ik},\vec{u}_{jl}\rangle}{\mu_{ik}+\kappa/2}(\frac{2}{\kappa}\tilde{\mu}_{jl}+1)+\epsilon_I\frac{1}{\tilde{\mu}_{ik}+\kappa/2}\frac{\tilde{\mu}_{jl}}{\tilde{\mu}_{jl}+\kappa/2}\\
        &<\frac{2}{\kappa}\epsilon_{\mu}\langle\vec{u}_{ik},\vec{u}_{jl}\rangle+\frac{2}{\kappa}\epsilon_I.
    \end{split}
\end{equation}
It is evident that for the two cases, we have, respectively, 
\begin{equation}
    \frac{2}{\kappa}\epsilon_{\mu}\langle\vec{u}_{ik},\vec{u}_{jl}\rangle+\frac{2}{\kappa}\epsilon_I\leq\xi\epsilon_1,
\end{equation}
and
\begin{equation}
    \frac{2}{\kappa}\epsilon_{\mu}\langle\vec{u}_{ik},\vec{u}_{jl}\rangle+\frac{2}{\kappa}\epsilon_I\leq\epsilon_1\epsilon_2.
\end{equation}
To summarize, the expected error bound of $R_{\kappa}$ entries has been proved.

\subsection{\label{app2.3}The error of \texorpdfstring{$J(Z)$}{}}
Now we estimate the error of $\det(R_{\kappa})$, when additive error of entries of $R_{\kappa}$ is bounded by $\epsilon$, we have $\|\tilde{R}_{\kappa}-R_{\kappa}\|_2\leq\epsilon d$ where $d$ is the dimension of $R_{\kappa}$. According to the inequality \cite{MatPerturb2013}: Let A and B be $d\times d$ matrices, thus it holds that
\begin{equation}
    |\frac{\det(A+B)-\det(A)}{\det(A)}|\leq \frac{d\mu_A\|B\|/\|A\|}{1-d\mu_A\|B\|/\|A\|},
    \label{eq30}
\end{equation}
when $d\mu_A\|B\|/\|A\|\leq1$, where $\mu_A$ is the condition number of $A$. To apply this inequality, it requires that $d^2\epsilon\mu_{R_{\kappa}}/\|R_{\kappa}\|_2\leq1$, where $d=\Theta(m\log N)$ and $\|R_{\kappa}\|_2/\mu_{R_{\kappa}}$ is exactly $\xi$. It suffices to choose $\epsilon<\frac{\xi}{2d^2}$. Substitute $A=R_{\kappa}$ and $B=R_{\kappa}-\tilde{R}_{\kappa}$ into Eq.~\ref{eq30}, we have
\begin{equation}
    |\frac{\det(\tilde{R}_{\kappa})-\det(R_{\kappa})}{\det(R_{\kappa})}|\leq\frac{d^2\epsilon\xi^{-1}}{1-d^2\epsilon\xi^{-1}}.
\end{equation}
Then relative error $\eta = |(\det(\tilde{R}_{\kappa})-\det(R_{\kappa}))/\det(R_{\kappa})|$ scales as $O(\xi^{-1}\epsilon m^2\log^2N)$. 

\section{\label{app3}Proofs of conclusions}
\subsection{\label{p4}Proof of Lemma~\ref{lem4}}

    We provide the effect of the circuit as proof. $O_{z_i}\otimes O_{z_i}$ implemented on the initial state $|i\rangle|0\rangle\otimes|j\rangle|0\rangle$ gives
    \begin{equation}
            (O_{z_i}\otimes O_{z_i})|j\rangle|0\rangle\otimes|k\rangle|0\rangle=|j\rangle|z_{ij}\rangle\otimes|k\rangle|z_{ik}\rangle.
    \end{equation}
    Next estimate kernel function in Eq.~\ref{eq5} by quantum arithmetics on a newly introduced quantum register initialized to $|z\rangle$ with $s=\lceil\log_2(1/\epsilon)\rceil$ qubits:
    \begin{equation}
        \begin{split}
            |j\rangle|z_{ij}\rangle&|k\rangle|z_{ik}\rangle|z\rangle\to|j\rangle|z_{ij}\rangle|k\rangle|z_{ik}\rangle|z\oplus K(z_{ij},z_{ik})\rangle.
        \end{split}
    \end{equation}
    Uncompute $z_{ij}$ and $z_{ik}$ by $O_{z_i}^{\dagger}$:
    \begin{equation}
        |j\rangle|0\rangle|k\rangle|0\rangle|z\oplus K(z_{ij},z_{ik})\rangle,
        \label{eq21}
    \end{equation}
    which achieves the goal of the lemma.
    Overall quantum arithmetics cost no more than $O(\mbox{poly}\log(N/\epsilon))$ quantum gates, and 2 queries to $O_{z_i}$ and $O^{\dagger}_{z_i}$, with precision of $K(z_{ij},z_{ik})$ up to $\epsilon$.

\subsection{\label{p4.5}Proof of Lemma~\ref{lem4.5}}

    First, we have
    \begin{equation}
        \begin{split}
        &(I_{2n}\otimes\textbf{CR})(\tilde{O}_{K_i}\otimes I_{1})(I_n\otimes H^{\otimes n}\otimes I_{s+1})|j\rangle|0^{\otimes s+1}\rangle\\
        &=\frac{1}{\sqrt{N}}\sum_{r=1}^N|j\rangle|r\rangle\bigg(|(\tilde{K}_{i})_{jr}\rangle((\tilde{K}_{i})_{jr}|0\rangle+\sqrt{1-(\tilde{K}_{i})_{jr}^2}|1\rangle\bigg),
        \end{split}
    \end{equation}
    and 
    \begin{equation}
        \begin{split}
            &(\textbf{SWAP}_n\otimes I_{s+1})(\tilde{O}_{K_i}\otimes I_{1})(I_n\otimes H^{\otimes n}\otimes I_{s+1})|k\rangle|0^{\otimes s+1}\rangle\\
            &=\frac{1}{\sqrt{N}}\sum_{t=1}^N|t\rangle|k\rangle|0^{\otimes s+1}\rangle.
        \end{split}
    \end{equation}
    Combining the two gives that
    \begin{equation}
        \begin{split}
            \langle k|\langle 0^{\otimes n+s+1}|\tilde{U}_{K_i}|j\rangle|0^{\otimes n+s+1}\rangle=\frac{(\tilde{K}_i)_{jk}}{N}=\frac{(\tilde{K}_i)_{kj}}{N}.
        \end{split}
    \end{equation}
    The controlled rotation can be implemented with the complexity of $O(\mbox{poly}\log(N/\epsilon))$ for the precision of rotation up to $\epsilon$. Therefore, an $(N,n+s+1,\epsilon)$-block-encoding of $\tilde{K}_i$ is given, which uses $O(\mbox{poly}\log(N/\epsilon))$ elementary gates and one query to $\tilde{O}_{K_i}$ and its inverse, as discussed in \cite{QSVT2019}. For our purpose, introduce another ancilla and then $\mbox{C}_{\Pi}\mbox{NOT}\cdot \tilde{U}_{K_i}\cdot \mbox{C}_{\Pi}\mbox{NOT}$ is an $(N,n+s+2,\epsilon)$-block-encoding of $K_i/N$, since with the projection,
    \begin{equation}
        \Pi=(I_n-|\vec{1}\rangle\langle\vec{1}|)\otimes|0^{\otimes n+s+1}\rangle\langle0^{\otimes n+s+1}|.
    \end{equation}
    we have
    \begin{equation}
        \Pi \tilde{U}_{K_i}\Pi=\frac{K_i}{N}\otimes|0^{\otimes n+s+1}\rangle\langle0^{\otimes n+s+1}|,
    \end{equation}
    where we use the relation that $K_i=(I_n-|\vec{1}\rangle\langle\vec{1}|)\tilde{K}_i(I_n-|\vec{1}\rangle\langle\vec{1}|)$.

\subsection{\label{p4.75}Proof of Lemma~\ref{lem4.75}}

Ignore the error induced by block-encoding process because of the $\mbox{poly}\log(N/\epsilon)$ dependence, the block-encoding of the $e^{iK_it}$ also has the eigenvalue $e^{i\lambda_{ik}t}$
for eigenvector $|u_{ik}\rangle$, in the subspace where the last $n+s+2$ qubits are set $|0^{\otimes n+s+1},1\rangle$. Therefore, the effect of the phase estimation of $e^{iK_it}$ is as claimed for $U_i$. The complexity comes mainly from the phase estimation, scaling as $\tilde{O}(1/\epsilon)$ for both query and gate one.

\subsection{\label{pl5}Proof of Lemma~\ref{lem5}}
The quantum circuit illustrated in Fig.~\ref{fig2} accomplishes the state preparation. We have, for the first part,
\begin{equation}
    \begin{split}
        (\bm{CR}\otimes I_{2n})&(I\otimes\tilde{O}_{K_i})|0^{\otimes 1+s}, 0^{\otimes s+2n}\rangle=\\
        &\sum_{j,k=1}^{N}(\tilde{K_i})_{jk}|0^{\otimes 1+s},j,k\rangle+|\perp_1\rangle.
    \end{split}
\end{equation}
The second part is a projection so that 
\begin{equation}
    \sum_{j,k=1}^{N}(\tilde{K_i})_{jk}|0^{\otimes 2}, j,k\rangle\to\sum_{j,k=1}^{N}(K_i)_{jk}|1^{\otimes 2}, j,k\rangle+|\perp_2\rangle.
\end{equation}
Represent $|K_i\rangle$ in the basis of eigenvectors of $K_i$ by 
\begin{equation}
    |K_i\rangle=\frac{1}{N}\sum_{k=1}^{N}\lambda_{ik}|u_{ik}\rangle|u_{ik}\rangle.
\end{equation}
Use the superscript to label the register, Eq.~\ref{eq28.9} comes up by
\begin{equation}
    \begin{split}
        |\phi_{\epsilon}\rangle=\Pi_{\epsilon/2}^3U_{i}^{13}&|K_i\rangle_{12}|0\rangle_3|0\rangle_4=\\
        &\sum_{k=1}^{N}\frac{\lambda_{ik}}{N}|u_{ik}\rangle_1|u_{ik}\rangle_2|\frac{\tilde{\lambda}_{ik}}{N}\rangle_3|0\rangle_4.
    \end{split}
\end{equation}
then the output of the quantum circuit in Fig.~\ref{fig2} is written by 
\begin{equation}
    \begin{split}
        |\psi_{\epsilon}\rangle\equiv &\Pi_{\epsilon/2}^4\Pi_{\epsilon/2}^3U_j^{24}U_{i}^{13}|K_i\rangle_{12}|0\rangle_3|0\rangle_4=\\
        &\sum_{k,l=1}^{N}\frac{\lambda_{ik}}{N}\langle u_{jl}|u_{ik}\rangle|u_{ik}\rangle_1|u_{jl}\rangle_2|\frac{\tilde{\lambda}_{ik}}{N}\rangle_3|\frac{\tilde{\lambda}_{jl}}{N}\rangle_4.
    \end{split}
    \label{eq28}
\end{equation}
Let the precision of phase estimation be $\epsilon/2$, and we have the claimed precision since eigenvalues larger than $\epsilon$ would not be discarded.
The complexity of the preparation comes mainly from $U_i,U_j$, scaling as $\tilde{O}(\frac{1}{\epsilon})$, which completes the proof.

\subsection{Proof of Theorem~\ref{the1}}
Note that $K_i$ is the same for $\tilde{y}_{ij}$ respectively in Eqs.~\ref{eq17} and~\ref{eq18}, so the distinction does not matter. Substitute $\bm{z}=\tilde{\bm{y}}'$ for Lemma~\ref{lem5}. Let $\epsilon=\frac{\xi\kappa\epsilon_1}{4}$ for $O_{\epsilon}$. Implement the amplitude estimation of $O_{\epsilon}$ with precision of also $\frac{\xi\kappa\epsilon_1}{4}$.
Directly compute $R_{\kappa}(Z)$ and $J(Z)$ with estimated $\lambda_{ik}$ and $\lambda_{ik}|\langle u_{jl}|u_{ik}\rangle|$. As shown in Appendix~\ref{app2.2} and~\ref{app2.3}, such estimation has the multiplicative error bounded above by $2d^2\epsilon_1$. Combined with the error induced by preprocessing in Lemma~\ref{lem11} with the precision of $\epsilon_2$, where it suffices to estimate the covariance matrix just once, here we accomplish the proof.

\subsection{Proof of Theorem~\ref{the2}}
Also do the substitution $\bm{z}=\tilde{\bm{y}}'$. Given an upper bound of $|\langle \vec{u}_{jl},\vec{u}_{ik}\rangle|$ for arbitrary $i,k,j,l$ by $\epsilon_2G$. $G$ is asymptotically a constant irrelevant to $N,\epsilon_1,\epsilon_2$ as in Appendix~\ref{app4}.
Let $\epsilon=\frac{\kappa\epsilon_1}{4G}$. Implement the amplitude estimation of $O_{\epsilon}$ with precision of $\frac{\epsilon_1\epsilon_2\kappa}{4}$. Note that the minimal eigenvalue $\xi$ of $R_{\kappa}(Z)$ is bounded below by $1-d\epsilon_2G$. Using results in Appendix~\ref{app2.2} and~\ref{app2.3} proves that the directly computed $J(Z)$ has multiplicative error bounded above by $O(d^2\epsilon_1\epsilon_2)$. Further consideration of preprocessing with the precision of $\epsilon_2$ completes the proof.

\subsection{\label{pl14}Proof of Lemma~\ref{lem14}}
Due to the factor $(\hat{I}-\hat{1}_i)$ of $\hat{K}$, we have
\begin{equation}
    \hat{1}_i\phi_{ik}=\int_{\mathbb{R}}\phi_{ik}(x)f_{z_i}(x)dx=0.
    \label{eq42}
\end{equation}
Up to the first order terms of $\epsilon$, we have first $F_{ij}=F_{ji}+O(\epsilon^2)$. The joint distribution of $z_i,z_j$ can be written as
\begin{equation}
    \begin{split}
        f_{z_i,z_j}&(x,y)=f_{z_i}(x)f_{z_j}(y)-\\
        &\epsilon\alpha_{ij}(y\frac{\partial}{\partial x}-x\frac{\partial}{\partial y})f_{s_i}(x)f_{s_j}(y)+O(\epsilon^2).
    \end{split}
\end{equation}
The distribution of $\bm{z}_i$ is almost invariant, namely $f_{x_i}(x)=f_{s_i}(x)+O(\epsilon^2)$. Considering Eq.\ref{eq42}, we have
\begin{equation}
    \begin{split}
        &M_{z_i,z_j}(\phi_{ik},\phi_{jl})=\int_{\mathbb{R}^2}\phi_{ik}(x)\phi_{jl}(y)[f_{z_i}(x)f_{z_j}(y)-\\
        &\epsilon\alpha_{ij}(y\frac{\partial}{\partial x}-x\frac{\partial}{\partial y})f_{s_i}(x)f_{s_j}(y)+O(\epsilon^2)]dxdy\\
        =\epsilon&\alpha_{ij}\int_{\mathbb{R}^2}\phi_{ik}(x)\phi_{jl}(y)(x\frac{\partial}{\partial y}-y\frac{\partial}{\partial x})f_{s_i}(x)f_{s_j}(y)dxdy+O(\epsilon^2).
    \end{split}
    \label{eq43}
\end{equation}

\subsection{\label{pl15}Proof of Lemma~\ref{lem15}}
That 
\begin{equation}
    \int_{\mathbb{R}^2}d\rho_{s_i}(x)d\rho_{s_j}(y)(x\frac{\partial}{\partial y}-y\frac{\partial}{\partial x})\phi_{ik}(x)\phi_{jl}(y)
\end{equation}
is always 0 for all $k,l$ such that $\mu_{ik},\mu_{jl}\neq0$ is equivalent to
\begin{equation}
    \int_{\mathbb{R}^2}d\rho_{s_i}(x)d\rho_{s_j}(y)(x\frac{\partial}{\partial y}-y\frac{\partial}{\partial x})K'_i(x,\alpha)K'_j(y,\beta)=0. 
    \label{eq45}
\end{equation}
Use the following notation:
\begin{align}
    &g_i(\alpha)=\int_{\mathbb{R}}d\rho_{s_i}(x)K(x,\alpha),\\
    &g'_i(\alpha)=\int_{\mathbb{R}}d\rho_{s_i}(x)\frac{\partial}{\partial x}K(x,\alpha),\\
    &C_i=\int_{\mathbb{R}}d\rho_{s_i}(\alpha)\alpha g_i(\alpha),
\end{align}
similarly for the definition of $g_j,C_j$. Omitting the terms that are 0 after integration or differential and substituting expression of Gaussian kernel (where we take $\sigma=1/\sqrt{2}$ for simplicity, little difference for other values of $\sigma$), Eq.\ref{eq45} is then
\begin{equation}
    \begin{split}
        2\beta g_j&(\beta)[\alpha g_i(\alpha)-\frac{1}{2}g'_i(\alpha)]-C_i g'_j(\beta)-\\
        &2\alpha g_i(\alpha)[\beta g_j(\beta)-\frac{1}{2}g'_j(\beta)]+C_j g'_i(\alpha)=0.
    \end{split}
\end{equation}
The above equation is simplified to
\begin{equation}
    [C_i-\beta g_j(\beta)]g'_i(\alpha)=[C_j-\alpha g_i(\alpha)]g'_j(\beta),
\end{equation} 
which leads to three solutions:
\begin{align}
    (1), g_i(\alpha)&=C_j/\alpha\\
    (2), g'_i(\alpha)&=g'_j(\beta)=0,\\
    \begin{split}
    (3), g'_i(\alpha)&+k[C_j-\alpha g_i(\alpha)]=0\to\\
    &g_i(\alpha)=e^{-a_1\alpha^2}(a_2+a_3\int_0^\alpha e^{a_1\alpha'^2}d\alpha').
    \end{split}
\end{align}
Conditions that $\int_{\mathbb{R}}f_{s_i}(x)dx=1$, $f_{s_i}(x)>0$ and $K(x,\alpha)=\exp\{-(x-\alpha)^2\}$ requires $\lim_{\alpha\to\infty}g_i(\alpha)=0$ so that the first two solutions should be discarded and $a_3=0$ for the last solution. $g_i(\alpha)=a_2e^{-a_1\alpha^2}$ leads to that $f_{s_i}$ must have the form of $\frac{\exp\{-x^2/\sigma_i^2\}}{\sqrt{2\pi}\sigma_i}$, i.e. $\bm{s}_i$ is Gaussian. The discussion above also works for $\bm{s}_j$ which completes the proof.

\section{\label{app4}Asymptotic properties}

In this section, we analyze the asymptotic properties of $K_i(Z)$ when $\bm{z}$ can be written as
\begin{equation}
    \bm{z}=(1+\epsilon_2F)\bm{s}+O(\epsilon^2).
    \label{eq29.5}
\end{equation}
In this setting, we show that the contrast function can only reach its extremum when the variables $z_i$ are independent if choosing the Gaussian kernel. This also helps prove our Theorem~\ref{the2}. We first introduce integral operators corresponding to Gram matrices used as a tool for analysis in Appendix~{app4.1}. In Appendix~\ref{app4.2}, an integral operator estimating $\langle\vec{u}_{jl},\vec{u}_{ik}\rangle$ is proposed and analyzed in the near-independent setting which gives the value of the inner product.

\subsection{\label{app4.1}Integral operators}
Let $K\in L^2(\mathbb{R}\times\mathbb{R})$ be a symmetric kernel and $f_{z_i}(x)\in L^2(\mathbb{R})$ be the PDF of a random variable $z_i$. For function $\phi(y)\in L^2(\mathbb{R})$, define several operators mapping from $L^2(\mathbb{R})$ to $L^2(\mathbb{R})$ as follows:
\begin{align}
    &\hat{T}_i\phi(y)=\int_{\mathbb{R}}K(x,y)\phi(x)d\rho_{z_i}(x),\\
    &\hat{I}\phi(y)=\phi(y),\\
    &\hat{1}_i\phi(y)=\int_{\mathbb{R}}\phi(x)d\rho_{z_i}(x),\\
    \label{eq33}&\hat{K}_i=(\hat{I}-\hat{1}_i)\hat{T}_i(\hat{I}-\hat{1}_i),
\end{align}
where $d\rho_{z_i}(x)$ refers to $f_{z_i}(x)dx$. In particular, use $K'(x,y)$ to write $\hat{K}_i$ concretely by
\begin{equation}
    \label{eq34}
    \hat{K}_i\phi(y)=\int_{\mathbb{R}}K'(x,y)\phi(x)d\rho_{z_i}(x).
\end{equation}
$\hat{K}_i$ is a Hilbert-Schmidt integral operator and then compact. Then we can talk about the $k$-th eigenvalues $\mu_{ik}$ and eigenfunctions $\phi_{ik}$ of $\hat{K}_i$, which is defined by
\begin{equation}
    \hat{K}_i\phi_{ik}(x)=\mu_{ik}\phi_{ik}(x). 
\end{equation}
Nyström method offers an estimate for this equation \cite{Numana2005}, by replacing the integration on $\mathbb{R}$ by summation on points sampled from the given distribution. For an arbitrary function $\phi$, we have
\begin{equation}
    \int_{\mathbb{R}}\phi(x)d\rho_{z_i}(x)\approx\frac{1}{N}\sum_{n=1}^N\phi(z_{in}),
\end{equation}
where $z_{in}$ are the $n$-th of all $N$ samples of $\bm{z}_i$. Existing analysis on convergency shows that with high probability, the integration can be approximated with precision up to $O(\frac{1}{\sqrt{N}})$ \cite{LearnKer2010}. Use the assumption given in Eq.~\ref{eq35new} for simplicity, then we have
\begin{equation}
    \mu_{ik}\phi_{ik}(z_{im})=(\hat{K}_i\phi_{ik})(z_{im})=\frac{1}{N}\sum_{n=1}^N(K_i(Z))_{mn}\phi_{ik}(z_{in}).
    \label{eqA11}
\end{equation}
Let $\Phi_{ik}=(\phi_{ik}(z_{i1}),\dots,\phi_{ik}(z_{iN}))^T$, then Eq.~\ref{eqA11} can be rewritten as
\begin{equation}
    \frac{K_i(Z)}{N}\Phi_{ik}=\mu_{ik}\Phi_{ik}.
    \label{eq39}
\end{equation}
This is properly the eigenequation of $K_i(Z)/N$, with eigenvectors $\Phi_{ik}$ which can be identified with $\vec{u}_{ik}$. Nyström method has established a connection between Gram matrices and integral operators from which properties of $K_i$ can be speculated. Conclusions come up in this way that eigenvalues of $K_i/N$ are about $\mu_{ik}$ irrelevant to $N$ and decay rapidly which allows a low-rank approximation \cite{KICA2002}.

\subsection{\label{app4.2}An operator for the inner product}
Define an operator $M_{z_i,z_j}$:
\begin{equation}
    M_{z_i,z_j}(\phi,\psi)=\int_{\mathbb{R}^2}\phi(x)\psi(y)f_{z_i,z_j}(x,y)dxdy,
    \label{eq40}
\end{equation}
where $f_{z_i,z_j}$ is the joint PDF of random variables $z_i,z_j$. It's direct to check that using Nyström method, $M_{z_i,z_j}(\phi_{ik},\phi_{jl})$ can be approximated by
\begin{equation}
    M_{z_i,z_j}(\phi_{ik},\phi_{jl})\approx\sum_{m=1}^N\phi_{ik}(z_{im})\phi_{jl}(z_{jm})=\langle \vec{u}_{ik},\vec{u}_{jl}\rangle.
\end{equation}
We give the expression of $M_{z_i,z_j}(\phi_{ik},\phi_{jl})$ when $\bm{z}$ satisfy Eq.~\ref{eq29.5} in Lemma~\ref{lem14}.
\begin{mylem}
    \label{lem14}
    Given $\epsilon_2^2\ll\epsilon_2$, suppose $\bm{z}_i$ satisfy Eq.\ref{eq29.5}. For eigenfunction $\phi_{ik}(x)$ defined in Eq.~\ref{eq34}, $M_{z_i,z_j}(\phi_{ik},\phi_{jl})$ equals
    \begin{equation}
    \epsilon_2F_{ij}\int_{\mathbb{R}^2}\phi_{ik}(x)\phi_{jl}(y)(x\frac{\partial}{\partial y}-y\frac{\partial}{\partial x})f_{z_i}(x)f_{z_j}(y)dxdy+O(\epsilon^2).
    \label{eqA12}
    \end{equation}
\end{mylem}
Denote by $C_{ik,jl}$ the integration in Eq.~\ref{eqA12} and observe the relation
\begin{equation}
    \begin{split}
        C_{ik,jl}\equiv&\int_{\mathbb{R}^2}\phi_{ik}(x)\phi_{jl}(y)(x\frac{\partial}{\partial y}-y\frac{\partial}{\partial x})f_{x_i}(x)f_{x_j}(y)dxdy\\
        =&-\int_{\mathbb{R}^2}d\rho_{s_i}(x)d\rho_{s_j}(y)(x\frac{\partial}{\partial y}-y\frac{\partial}{\partial x})\phi_{ik}(x)\phi_{jl}(y).
    \end{split}
    \label{eq44}
\end{equation}
This relation helps prove that if and only if $\bm{s}_i,\bm{s}_j$ are both Gaussian, $C_{ik,jl}=0$ holds for $k,l$ that satisfy $\mu_{ik},\mu_{jl}\neq0$, as stated formally in Lemma~\ref{lem15}.
\begin{mylem}
    \label{lem15}
    With the same setting of Lemma~\ref{lem14} and given PDFs $f_{s_i}$ of $\bm{s}_i$, if Gaussian kernel is selected, $K(x,y)=\exp\{(x-y)^2/\sigma^2\}$, $C_{ik,jl}=0$ holds for fixed $i,j$ and all $k,l$ such that $\mu_{ik},\mu_{jl}\neq0$ if and only if $\bm{s}_i,\bm{s}_j$ are both Gaussian variables.
\end{mylem}
Lemma.~\ref{lem14} and~\ref{lem15} supports that the Gaussian kernel is an appropriate choice since at least in the near-independent setting, $R_{\kappa}$ behaves similarly to the mixing coefficients $F$ with such choice, except for Gaussian variables which cannot be separated by ICA inherently \cite{ICAbook2010}. The two proofs are left to Appendix~\ref{pl14} and~\ref{pl15}. Moreover, we define $D_{ik,jl}$ as
\begin{equation}
    D_{ik,jl}^2\equiv\int_{\mathbb{R}^2}d\rho_{s_i}(x)d\rho_{s_j}(y)\big[(x\frac{\partial}{\partial y}-y\frac{\partial}{\partial x})\phi_{ik}(x)\phi_{jl}(y)\big]^2-C_{ik,jl}^2,
\end{equation}
which is $N$ times the variance of estimating $C_{ik,jl}$ using N samples \cite{MonteCarlo1998}. Combining Lemma~\ref{lem14} and~\ref{lem15} gives the value of $\langle \vec{u}_{ik},\vec{u}_{jl}\rangle$. 
\begin{mycor}[Approximation of inner product]
    \label{cor6}
    With the same setting of Lemma~\ref{lem14}, it holds that
    \begin{equation}
        |\langle\vec{u}_{ik},\vec{u}_{jl}\rangle-\epsilon_2F_{ij}C_{ik,jl}|<\delta\epsilon_2F_{ij}\frac{D_{ik,jl}}{\sqrt{N}},
    \end{equation}
    with at least probability $1-1/\delta^2$.
\end{mycor}
It is direct by using the Chebyshev inequality.

\bibliography{ref.bib}
\end{document}